\documentclass[12pt]{article}%

\usepackage[cm]{fullpage}

\usepackage[usenames,dvipsnames,svgnames,table]{xcolor}
\usepackage{bm}%
\usepackage{paralist}%
\usepackage{picins}%
\usepackage{mleftright}%
\usepackage{tabto}%
\usepackage{xifthen}%
\usepackage{stmaryrd}%

\usepackage{xspace}%
\usepackage{ifpdf}%
\usepackage{graphicx}%
\usepackage[normalem]{ulem} %
\usepackage{caption}%
\usepackage{amsmath}%
\usepackage{amssymb}%
\usepackage[amsmath,thmmarks]{ntheorem}%
\theoremseparator{.}%
\usepackage{wrapfig}

\usepackage{euscript}
\usepackage{mathrsfs}

\usepackage{verbatim}%

\usepackage{bbding}%

\usepackage{bbm}%
\usepackage{bbding}

   \usepackage{ifluatex}
   \usepackage{ifxetex}

   \ifluatex 
      \usepackage{fontspec}
      \usepackage[utf8]{luainputenc}
   \else
       \ifxetex 
          \usepackage{fontspec}
       \else
          \usepackage[T1]{fontenc}
          \usepackage[utf8]{inputenc}
       \fi
   \fi

\providecommand{\BibLatexMode}[1]{}
\providecommand{\BibTexMode}[1]{#1}

\ifx\UseBibLatex\undefined%
  \renewcommand{\BibLatexMode}[1]{}
  \renewcommand{\BibTexMode}[1]{#1}
\else
  \renewcommand{\BibLatexMode}[1]{#1}
  \renewcommand{\BibTexMode}[1]{}
\fi

\BibLatexMode{%
   \usepackage[bibencoding=ascii,style=alphabetic,backend=biber]{biblatex}%
   \usepackage{sariel_biblatex}%
}

\newcommand{\kentdelete}[1]{}

\newcommand{\IfPrinterVer}[2]{#2}%
\providecommand{\Mh}[1]{{#1}}%

\IfFileExists{.latex_printer_friendly}{\def\GenPrinterVer{1}}{}%

\ifx\GenPrinterVer\undefined
   \IfFileExists{.latex_color}{\def\GenColorMath{1}}{}
\else
   \renewcommand{\IfPrinterVer}[2]{#1}%
\fi

\ifx\GenColorMath\undefined \else %
\renewcommand{\Mh}[1]{{\textcolor{red}{#1}}}%
\fi

\usepackage{secdot}%
\sectiondot{subsection}%
\sectiondot{subsubsection}%

\IfPrinterVer{%
   \usepackage{hyperref}%
}{%
   \usepackage{hyperref}%
   \hypersetup{%
      breaklinks,%
      ocgcolorlinks, colorlinks=true,%
      urlcolor=[rgb]{0.25,0.0,0.0},%
      linkcolor=[rgb]{0.5,0.0,0.0},%
      citecolor=[rgb]{0,0.2,0.445},%
      filecolor=[rgb]{0,0,0.4},
      anchorcolor=[rgb]={0.0,0.1,0.2}%
   }
}

\newcommand{\etal}{\textit{et~al.}\xspace}
\newcommand{\Hastad}{H\r{a}stad\xspace}

\newlength{\savedparindent}
\newcommand{\SaveIndent}{\setlength{\savedparindent}{\parindent}}
\newcommand{\RestoreIndent}{\setlength{\parindent}{\savedparindent}}

\newcommand{\Term}[1]{\textsf{#1}}

\definecolor{blue25}{rgb}{0, 0, 11}
\newcommand{\emphic}[2]{%
  \textcolor{blue25}{%
    \textbf{\emph{#1}}}%
  \index{#2}}

\ifx\PDFBlackWhite\undefined
\else
\renewcommand{\emphic}[2]{\textbf{\emph{#1}}}
\fi

\newcommand{\emphi}[1]{\emphic{#1}{#1}}

\newcommand{\cardin}[1]{\left| {#1} \right|}%
\newcommand{\ceil}[1]{\left\lceil {#1} \right\rceil}

\newcommand{\pth}[1]{\mleft({#1}\mright)}

\newcommand{\sepw}[1]{\left|\, {#1} \right.}

\newcommand{\brc}[1]{\left\{ {#1} \right\}}
\newcommand{\setof}[1]{\left\{ {#1} \right\}}
\newcommand{\Set}[2]{\left\{ #1 \;\middle\vert\; #2 \right\}}
\newcommand{\pbrc}[1]{\mleft[ {#1} \mright]}
\newcommand{\SetDiff}{\triangle}

\newcommand{\Ex}[1]{\mathop{\mathbf{E}}\pbrc{#1}}

\newcommand{\Prob}[1]{\mathop{\mathbf{Pr}}\pbrc{#1}}

\newcommand{\remove}[1]{}

\newtheorem{theorem}{Theorem}
\newtheorem{lemma}[theorem]{Lemma}%
\newtheorem*{restate*}[theorem]{Restatement of }%
\newtheorem{corollary}[theorem]{Corollary}%

\newtheorem{observation}[theorem]{Observation}%

\newcommand{\myqedsymbol}{\rule{2mm}{2mm}}

\theoremstyle{remark}%
\theoremheaderfont{\sf}%
\theorembodyfont{\upshape}%

\newtheorem{defn}[theorem]{Definition}
\newtheorem*{defn:unnumbered}{Definition}%

\newtheorem*{remark:unnumbered}[theorem]{Remark}%
\newtheorem{remark}[theorem]{Remark}%
\newtheorem{example}[theorem]{Example}%

\theoremheaderfont{\em}%
\theorembodyfont{\upshape}%
\theoremstyle{nonumberplain}%
\theoremseparator{}%
\theoremsymbol{\myqedsymbol}%
\newtheorem{proof}{Proof:}%

\numberwithin{figure}{section}%
\numberwithin{table}{section}%
\numberwithin{equation}{section}%

\newcommand{\HLinkSuffix}[3]{\hyperref[#2]{#1\ref*{#2}{#3}}}
\newcommand{\HLinkShort}[2]{\hyperref[#2]{#1\ref*{#2}}}
\newcommand{\HLink}[2]{\hyperref[#2]{#1~\ref*{#2}}}
\newcommand{\HLinkPage}[2]{\hyperref[#2]{#1~\ref*{#2}%
    $_\text{p\pageref{#2}}$}}

\newcommand{\eqrefpar}[1]{\hyperref[equation:#1]{(\ref*{equation:#1})}} %

\newcommand{\figlab}[1]{\label{fig:#1}}
\newcommand{\figref}[1]{\HLink{Figure}{fig:#1}}

\newcommand{\seclab}[1]{\label{sec:#1}} %
\newcommand{\secref}[1]{\HLink{Section}{sec:#1}} %

\newcommand{\corlab}[1]{\label{cor:#1}}
\newcommand{\corref}[1]{\HLink{Corollary}{cor:#1}}%

\providecommand{\deflab}[1]{\label{def:#1}}
\newcommand{\defref}[1]{\HLink{Definition}{def:#1}}

\newcommand{\exmlab}[1]{\label{example:#1}}
\newcommand{\exmref}[1]{\HLink{Example}{example:#1}}

\newcommand{\lemlab}[1]{\label{lemma:#1}}
\newcommand{\lemref}[1]{\HLink{Lemma}{lemma:#1}}

\newcommand{\lemrefshort}[1]{\HLinkShort{L}{lemma:#1}}

\newcommand{\itemlab}[1]{\label{item:#1}}
\newcommand{\itemref}[1]{\HLinkSuffix{(}{item:#1}{)}}

\newcommand{\remlab}[1]{\label{rem:#1}}
\newcommand{\remref}[1]{\HLink{Remark}{rem:#1}}

\newcommand{\obslab}[1]{\label{observation:#1}}
\newcommand{\obsref}[1]{\HLink{Observation}{observation:#1}}

\newcommand{\thmlab}[1]{{\label{theo:#1}}}
\newcommand{\thmref}[1]{\HLink{Theorem}{theo:#1}}

\newcommand{\ETH}{\Term{ETH}\xspace}
\newcommand{\SETH}{\Term{SETH}\xspace}
\newcommand{\TrSAT}{\ProblemC{$3$SAT}\xspace}
\newcommand{\kSAT}{\ProblemC{$k$SAT}\xspace}

\providecommand{\Mh}[1]{{#1}}

\newcommand{\IntRange}[1]{\left\llbracket #1 \right\rrbracket}

\renewcommand{\th}{th\xspace}

\newcommand{\obj}{\Mh{f}}%
\newcommand{\objA}{\Mh{g}}%
\newcommand{\objB}{\Mh{h}}%
\newcommand{\ds}{\displaystyle}%

\newcommand{\ObjSet}{{\Mh{\mathcal{U}}}}%
\newcommand{\ObjSetA}{\Mh{\mathcal{V}}}%
\newcommand{\ObjSetB}{\Mh{\mathcal{H}}}%

\newcommand{\Cover}{\Mh{\mathcal{C}}}%

\newcommand{\bNotation}[1]{\Mh{#1}}%
\newcommand{\ball}{\bNotation{b}}%
\newcommand{\ballA}{\bNotation{b'}}%
\newcommand{\ballY}[2]{\Mh{\mathbbm{b}}\pth{#1, #2}}%
\newcommand{\BallSet}{\Mh{\mathcal{B}}}%

\newcommand{\SetA}{\Mh{X}}%
\newcommand{\SetB}{\Mh{Y}}%
\newcommand{\SetC}{\Mh{U}}

\newcommand{\TriSet}{\Mh{\EuScript{T}}}
\newcommand{\cen}{\Mh{c}}

\newcommand{\optFSet}{\Mh{\mathcal{O}}} %
\newcommand{\locFSet}{\Mh{\mathcal{L}}} %
\newcommand{\optSet}{\Mh{O}}%

\newcommand{\flower}{\Mh{F}}%

\newcommand{\flowerX}[1]{\flower_{#1}}

\newcommand{\optFl}[1]{\flowerX{#1}}%
\newcommand{\locFl}[1]{\flowerX{#1}'}%

\newcommand{\bdDiv}{\Mh{\mathcal{B}}} %

\newcommand{\lpnt}{\Mh{l}}%
\newcommand{\opnt}{\Mh{o}}%

\newcommand{\dblC}{\Mh{c_{\mathrm{dbl}}}}
\newcommand{\dblCd}{\Mh{{c}_d}}
\newcommand{\Weight}{\Mh{W}}%
\newcommand{\weightOp}{\operatorname{\Mh{w}}}
\newcommand{\weightX}[1]{\weightOp\pth{#1}}

\newcommand{\excess}{\Mh{\mathcal{E}}}%

\newcommand{\volX}[1]{\operatorname{vol}\pth{#1}}
\newcommand{\diamX}[1]{\operatorname{\Mh{diam}}\pth{#1}}%

\newcommand{\Family}{\Mh{\EuScript{F}}}%
\newcommand{\FamilyA}{\Mh{\EuScript{G}}}%

\newcommand{\sphereC}{{\mathbb{{S}}}}%
\newcommand{\sphereCBig}{\mathbb{S}}%
\newcommand{\sphereX}[2]{\sphereCBig\pth{#1, #2}}%

\renewcommand{\Re}{{\mathbb{R}}}

\newcommand{\naturalnumbers}{\mathbb{N}} %

\newcommand{\distCharX}[1]{\mathsf{d}_{#1}}

\newcommand{\repX}[1]{\mathrm{rep}\pth{#1}}

\newcommand{\gradC}{\mathbbm{d}} %

\newcommand{\gradY}[2]{\gradC_{#1}\pth{#2}}

\providecommand{\lexprod}{\bullet} %
\newcommand{\edgedensityof}[1]{\frac{\cardin{\EdgesX{#1}}}{\cardin{\verticesof{#1}}}}

\newcommand{\centerX}[1]{\mathrm{center}\pth{#1}}
\newcommand{\FDecomp}[2]{%
  \raisebox{-2pt}{\text{\small \SixFlowerPetalDotted}}%
  \pth{#1, #2}}

\newcommand{\TreeX}[1]{T_{#1}}
\newcommand{\nnK}[3]{\Mh{\mathrm{NN}}_{#1}\pth{#2, #3}}
\newcommand{\ServeX}[1]{\Mh{S}\pth{#1}}%

\newcommand{\MaxDemand}{{\widehat{\delta}}}

\newcommand{\reachX}[1]{\Mh{\tau}\pth{#1}}%
\newcommand{\MaxReach}{{\widehat{\tau}}}%

\newcommand{\clusters}{\Mh{\mathcal{W}}} %
\newcommand{\cluster}{\Mh{C}} %

\newcommand{\clusterA}{\Mh{\mathcal{C}}} %

\newcommand{\clusterX}[1]{\cluster\pth{#1}}

\providecommand{\excessCmd}{\operatorname{excess}} %
\newcommand{\excessof}[1]{\excessCmd\pth{#1}} %

\newcommand{\trA}{\Mh{\pi}}
\newcommand{\trB}{\Mh{\sigma}}

\newcommand{\distY}[2]{\mleft\| #1 - #2 \mright\|}
\newcommand{\normX}[1]{\left\| #1 \right\|}

\newcommand{\Prop}{\Mh{\Pi}}%
\providecommand{\CNFX}[1]{ {\em{\textrm{(#1)}}}}%
\providecommand{\CNFCCCG}{\CNFX{CCCG}}%
\newcommand{\cDensity}{\Mh{\rho}} %
\newcommand{\densityOp}{\Mh{\mathop{\mathrm{density}}}}%
\newcommand{\densityX}[1]{\densityOp\pth{#1}}%
\newcommand{\NbrX}[1]{\Mh{N}\pth{#1}}

\newcommand{\PntSet}{\ensuremath{\Mh{P}}\xspace}%
\newcommand{\PntSetA}{\ensuremath{\Mh{Q}}\xspace}%

\newcommand{\PointDec}[1]{\Mh{#1}}

\newcommand{\pnt}{\PointDec{p}}%
\newcommand{\pntA}{\PointDec{q}}%
\newcommand{\eps}{\Mh{\varepsilon}}%

\newcommand{\SepSet}{\Mh{Z}}%

\newcommand{\DomSet}{\Mh{D}}

\newcommand{\CovSet}{\Mh{R}} %
\newcommand{\DiskOrg}{\mathsf{disk}}
\newcommand{\CH}{\mathcal{CH}}

\newcommand{\Metric}{\mathcal{X}}

\newcommand{\Vertices}{\Mh{V}}%

\newcommand{\SetX}{\Mh{X}}
\newcommand{\SetY}{\Mh{Y}}
\newcommand{\SetZ}{\Mh{Z}}

\newcommand{\VerticesX}[1]{\Mh{V}\pth{#1}}%
\newcommand{\verticesof}[1]{\Mh{V}\pth{#1}}%

\newcommand{\Edges}{\Mh{E}}
\newcommand{\EdgesX}[1]{\Edges\pth{#1}}

\newcommand{\optset}{\Mh{O}} %
\newcommand{\locset}{\Mh{L}} %

\newcommand{\edgeY}[2]{{#1#2}}

\newcommand{\Opt}{\Mh{{O}}}%
\newcommand{\locSol}{\Mh{{L}}}%

\newcommand{\BVertices}{\Mh{B}}%

\newcommand{\bSize}{\Mh{{b}}} %
\newcommand{\oSize}{\Mh{{o}}} %
\newcommand{\lSize}{\Mh{{l}}} %

\newcommand{\iCov}{\Mh{\omega}}%
\newcommand{\ICovGraph}[2]{#1\pbrc{#2}}%

\newcommand{\icgA}{\ICovGraph{\graph}{\Family}} %
\newcommand{\IObjSet}[2]{#1\pbrc{#2}}

\newcommand{\ProblemC}[1]{\textsf{#1}}

\DefineNamedColor{named}{ColorComplexityClass}{cmyk}{0.64,0.0,0.95,0.80}
\providecommand{\ComplexityClass}[1]{%
  {{\textcolor[named]{ColorComplexityClass}{%
        \textsc{#1}}}}}

\newcommand{\Interval}{J}

\newcommand{\POLYT}{\ComplexityClass{P}\xspace}

\newcommand{\poly}{\operatorname{poly}}%
\newcommand{\polylog}{\operatorname{polylog}}
\newcommand{\MaxSNPHard}{\ComplexityClass{MaxSNP-Hard}\xspace}

\newcommand{\PTAS}{\Term{PTAS}\xspace}
\newcommand{\QPTAS}{\Term{QPTAS}\xspace}

\newcommand{\NP}{\ComplexityClass{NP}\xspace}
\providecommand{\NPComplete}{%
  {\ComplexityClass{NP-Complete}}%
  \index{NP!Complete}\xspace%
}
\providecommand{\NPHard}{{\ComplexityClass{NP-Hard}}%
  \index{NP!hard}\xspace}
\newcommand{\APXHard}{\ComplexityClass{APX-Hard}\xspace}

\newcommand{\Nesetril}{N{e{\v s}et{\v r}il}\xspace}
\newcommand{\si}[1]{#1}%

\newcommand{\IGraph}[1]{\graph_{#1}}
\newcommand{\GInduced}[1]{\graph_{|{#1}}}

\newcommand{\Partition}{\Mh{\mathcal{P}}}%
\newcommand{\atgen}{\symbol{'100}}%

\newcommand{\KentThanks}%
{%
  \thanks{%
    Department of Computer Science; %
    University of Illinois; %
    201 N. Goodwin Avenue; %
    Urbana, IL, 61801, USA; %
    {\tt quanrud2\atgen{}illinois.edu}; %
    {\tt\href%
      {http://illinois.edu/\string~quanrud2/}%
      {http://illinois.edu/\string~quanrud2/}%
    }%
    . %
  }%
}%
\newcommand{\SarielThanks}[1][]{%
  \thanks{Department of Computer Science; %
    University of Illinois; %
    201 N. Goodwin Avenue; %
    Urbana, IL, 61801, USA; %
    {\tt sariel\atgen{}illinois.edu}; %
    {\tt \url{http://sarielhp.org/}}. #1}} %

 \newcommand{\nfrac}[2]{#1/#2}
\newcommand{\demandOp}{\Mh{\delta}}%
\newcommand{\demandX}[1]{\demandOp\pth{#1}}%

\DeclareMathAlphabet{\mathantt}{OT1}{antt}{li}{it}
\DeclareMathAlphabet{\mathpzc}{OT1}{pzc}{m}{it}

\DeclareMathAlphabet{\mathcalligra}{T1}{calligra}{m}{n}

\newcommand{\lenX}[1]{\normX{#1}}

\newcommand{\headsX}[1]{\Mh{\mathrm{heads}}\pth{#1}}

\newcommand{\exSize}{\Mh{\lambda}}%

\newcommand{\IncludeGraphics}[2][]{%
  \typeout{}%
  \typeout{Graphics: #2}%
  \typeout{\ includegraphics[#1]{#2}}%
  \includegraphics[#1]{#2}
  \typeout{}%
}

\newcommand{\IncGraphPage}[4][]{%
  \IncludeGraphics[page=#4,#1]{{#2/#3}}
}

\newcommand{\Instance}{I}

\newcommand{\defGraph}{\graph = (\Vertices,\Edges)}

\newcommand{\class}{\Mh{\mathcal{C}}}
\newcommand{\minorsDY}[2]{{\Mh{\nabla}_{\!#1}} \pth{#2} } %

\newcommand{\GraphNotation}[1]{\Mh{#1}}

\newcommand{\clusterZ}[1]{\Mh{C}_{#1}}
\newcommand{\cvX}[1]{\Mh{c}_{#1}} %
\newcommand{\prmtX}[1]{\Mh{\pi}\pth{#1}}%
\newcommand{\clusteredge}[2]{\edgeY{\clusterZ{#1}}{\clusterZ{#2}}}
\newcommand{\scoopedge}[2]{\edgeY{\scoop{#1}}{\scoop{#2}}}
\newcommand{\rcvX}[1]{\cvX{\prmtX{#1}}}%
\newcommand{\randomcluster}[1]{\clusterZ{\prmtX{#1}}}%
\newcommand{\scoop}[1]{\clusterZ{#1}'}
\newcommand{\rsX}[1]{\Mh{C}_{\!\prmtX{#1}}'}%
\newcommand{\scoops}{\Family'} %
\newcommand{\smalledges}{\Edges_1} %
\newcommand{\bigedges}{\Edges_2} %

\newcommand{\graph}{\GraphNotation{G}}%
\newcommand{\graphA}{\GraphNotation{H}}%
\newcommand{\graphB}{\GraphNotation{K}}%
\newcommand{\TheoremDefExt}[3]{%
  \expandafter\newcommand\csname bodyThm#1\endcsname{#2}
  \ifthenelse{\isempty{#3}}{%
    \begin{theorem}%
      \thmlab{#1}%
      #2
    \end{theorem}%
  }{%
    \begin{theorem}[#3]%
      \thmlab{#1}%
      #2
    \end{theorem}%
  }%
}

\newcommand{\TheoremBody}[1]{%
  \csname bodyThm#1\endcsname%
}

\newcommand{\LemmaDefExt}[3]{%
  \expandafter\newcommand\csname bodyLemBody#1\endcsname{#2}
  \ifthenelse{\isempty{#3}}{%
    \begin{lemma}%
      \lemlab{#1}%
      #2
    \end{lemma}%
  }{%
    \begin{lemma}[#3]%
      \lemlab{#1}%
      #2
    \end{lemma}%
  }%
}

\newcommand{\LemmaBody}[1]{%
  \csname bodyLemBody#1\endcsname%
}

\newcommand{\mytfrac}[2]{%
   \frac{#1\rule[-3\lineskip]{-0.001cm}{0\lineskip}}%
   {#2 \rule{-0.001cm}{7\lineskip}}}

\definecolor{sarielChangeColor}{rgb}{0.3,0.451,0.055}%

\newcommand{\ProblemE}[3]{%
  \begin{minipage}{0.31\linewidth}
    \smallskip%
    #1%
    \smallskip%
  \end{minipage}
  &
  \begin{minipage}{0.26\linewidth}
    \smallskip%
    #2%
    \smallskip%
  \end{minipage}
  &
  \begin{minipage}{0.3\linewidth}
    \smallskip%
    #3%
    \smallskip%
  \end{minipage}%
  \\%
}

\IfFileExists{sariel_computer.sty}{\def\sarielComp{1}}{}
\ifx\sarielComp\undefined%
\newcommand{\SarielComp}[1]{}
\newcommand{\NotSarielComp}[1]{#1}%
\else
\newcommand{\SarielComp}[1]{#1}%
\newcommand{\NotSarielComp}[1]{}%
\fi

\SarielComp{%
  \ifx\colorMath\undefined%
  \else
  \DefineNamedColor{named}{ColorMath}{rgb}{0.5,0.2,0}
  \renewcommand{\Mh}[1]{{\textcolor{ColorMath}{#1}}}
  \fi
}

\IfPrinterVer{%
  \definecolor{blue25}{rgb}{0,0,0}
  \DefineNamedColor{named}{RedViolet} {rgb}{0,0,0}%
  \DefineNamedColor{named}{ChangeColor}{rgb}{0, 0, 0}
}{%
  \definecolor{blue25}{rgb}{0,0,0.7}
  \DefineNamedColor{named}{RedViolet} {cmyk}{0.07,0.90,0,0.34}
  \DefineNamedColor{named}{ChangeColor}{rgb}{0, 0.2, 0.0}
}%

\BibLatexMode{%
}

\begin{document}

\title{Approximation Algorithms for Polynomial-Expansion and
  Low-Density Graphs%
  \thanks{Work on this paper was partially supported by a NSF AF
    awards CCF-1421231, and 
    CCF-1217462. %
    A preliminary version of this paper appeared in E{S}A 2015
    \cite{hq-aapel-15}. A talk by the first author in SODA 2016 was
    based to some extent on the work in this paper.}%
}%

\author{%
  Sariel Har-Peled%
  \SarielThanks{}%
  \and%
  Kent Quanrud%
  \KentThanks{}%
}%

\maketitle



\begin{abstract}
  We investigate the family of intersection graphs of low density
  objects in low dimensional Euclidean space.  This family is quite
  general, includes planar graphs, and in particular is a subset of
  the family of graphs that have polynomial expansion.

  We present efficient $(1+\eps)$-approximation algorithms for
  polynomial expansion graphs, for \ProblemC{Independent Set},
  \ProblemC{Set Cover}, and \ProblemC{Dominating Set} problems, among
  others, and these results seem to be new. Naturally, \PTAS{}'s for
  these problems are known for sub{}classes of this graph family.
  These results have immediate applications in the geometric
  domain. For example, the new algorithms yield the only \PTAS known
  for covering points by fat triangles (that are shallow).

  We also prove corresponding hardness of approximation for some of
  these optimization problems, characterizing their intractability
  with respect to density. For example, we show that there is no \PTAS
  for covering points by fat triangles if they are not shallow, thus
  matching our \PTAS for this problem with respect to depth.
\end{abstract}

\section{Introduction}

\paragraph{Motivation.}

Geometric set cover, as the name suggests, is a geometric
instantiation of the classical set cover problem. For example, given a
set of points and a set of triangles (with fixed locations) in the
plane, we want to select a minimum number of triangles that cover all
of the given points. Similar geometric variants can be defined for
independent set, hitting set, dominating set, and the like.

For relatively simple shapes, such geometric instances should be
computationally easier than the general problem. By now there is a
large yet incomplete collection of results on such problems, listed
below. For example, one can get $(1+\eps)$-approximation to the
geometric set cover problem when the regions are disks, but we do not
have such an approximation algorithm if the regions are fat triangles
of similar size.  This discrepancy seems arbitrary and somewhat
baffling, and the purpose of this work is to better understand these
subtle gaps.

\paragraph{Plan of attack.}
We explore the type of graphs that arises out of these geometric
instances, and in the process introduce the class of low-density
graphs. We explore the properties of this graph class, and the
optimization problems that can be approximated efficiently on the
broader class of graphs that have small separators. Separabilitity
turns out to be the key ingredient needed for efficient approximation.
We also study lower bounds for such instances, characterizing when
they are computationally hard.

\begin{figure}[t]
  \begin{center}
    \begin{tabular}{|l|%
      l|l|}
      \hline
      Objects & \si{Approx.} \si{Alg.} & Hardness\\
      \hline%
      \ProblemE{Disks/pseudo-disks}%
      { \PTAS \cite{mrr-sahsg-14}}%
      { Exact version \NPHard\\ \cite{fg-oafac-88} }
      \hline%
      \ProblemE{Fat triangles of same size}%
      {$O(1)$  \cite{cv-iaags-07}}%
      { \APXHard: \lemref{no:PTAS:fat:tr:set:cover}\\
      I.e., no \PTAS possible.
      }
      \hline%
      \ProblemE{Fat objects in $\Re^2$}%
      {$O(\log^* \mathrm{opt})$  \cite{abes-ibulf-14}}%
      {\APXHard:  \lemrefshort{no:PTAS:fat:tr:set:cover}}
      \hline%
      \ProblemE{Objects $\subseteq \Re^d$, $O(1)$
      density\\
      E.g. fat objects, $O(1)$ depth.}%
      {\PTAS: \thmref{g:hitting:set:cover}}%
      {Exact version \NPHard\\%
      \cite{fg-oafac-88}}
      \hline%
      \ProblemE{Objects with $\polylog$ density}%
      {\QPTAS: \thmref{g:hitting:set:cover}}
      {No \PTAS under \ETH\\
      \lemref{e:t:h:g:polylog:d}%
      }%
      \hline%
      \ProblemE{Objects with  density $\cDensity$ in $\Re^d$}%
      {%
      \PTAS:
      \thmref{g:hitting:set:cover}
      \\
      RT:
      $n^{O(\cDensity^{(d+1)/d}/\epsilon^d)}$.%
      }{%
      No $(1+\eps)$-approx

      with RT
      $n^{\poly(\log \cDensity , 1/\eps)}$

      \si{assuming} \ETH: 
      \lemrefshort{e:t:h:g:polylog:d}%
      }
      \hline
    \end{tabular}\hfill
  \end{center}

  \vspace{-0.3cm}
  \caption{Known results about the complexity of geometric
    set-cover. The input consists of a set of points and a set of
    objects, and the task is to find the smallest subset of objects
    that covers the points.  To see that the hardness proof of Feder
    and Greene \cite{fg-oafac-88} indeed implies the above, one just
    needs to verify that the input instance their proof generates has
    bounded depth.  A \QPTAS is an algorithm with running time
    $n^{O(\poly(\log n,1/\eps))}$.  }
  \figlab{set:cover:summary}
\end{figure}

\subsection{Background}

\subsubsection{Optimization problems}

\paragraph{Independent set.}

Given an undirected graph $\defGraph$, an \emph{independent set} is a
set of vertices $\SetA \subseteq \Vertices$ such that no two vertices
in $\SetA$ are connected by an edge. It is \NPComplete to decide if a
graph contains an independent set of size $k$ \cite{k-racp-72}, and
one cannot approximate the size of the maximum independent set to
within a factor of $n^{1-\eps}$, for any fixed $\eps > 0$, unless
$\POLYT = \NP$ \cite{h-chaw-99}.

\paragraph{Dominating set.}

Given an undirected graph $\defGraph$, a \emph{dominating set} is a
set of vertices $\DomSet \subseteq \Vertices$ such that every vertex
in $\graph$ is either in $\DomSet$ or adjacent to a vertex in
$\DomSet$. It is \NPComplete to decide if a graph contains a
dominating set of size $k$ (by a simple reduction from set cover,
which is \NPComplete \cite{k-racp-72}), and one cannot obtain a
$c \log n$ approximation (for some constant $c$) unless $\POLYT = \NP$
\cite{rs-sbepl-97}.

\subsubsection{Graph classes}
\paragraph{Density.}

Informally, a set of objects in $\Re^d$ is \emph{low-density} if no
object can intersect too many objects that are larger than it. This
notion was introduced by van \si{der} Stappen
\etal~\cite{sobv-mpelo-98}, although weaker notions involving a single
resolution were studied earlier (e.g. in the work by Schwartz and
Sharir \cite{ss-empae-85}). A closely related geometric property to
density is \emph{fatness}. Informally, an object is fat if it contains
a ball, and is contained inside another ball, that up to constant
scaling are of the same size.  Fat objects have low union complexity
\cite{aps-sugo-08}, and in particular, shallow fat objects have low
density \cite{f-mpafo-92}. Here, a set of shapes is \emph{shallow} if
no point is covered by too many of them.

\paragraph{Intersection graphs.}

A set $\Family$ of objects in $\Re^d$ induces an \emph{intersection
  graph} $\IGraph{\Family}$ having $\Family$ as its set of vertices,
and two objects $\obj, \objA \in \Family$ are connected by an edge if
and only if $\obj \cap \objA \neq \emptyset$.  Without any
restrictions, intersection graphs can represent any graph. Motivated
by the notion of density, a graph is a \emph{low-density graph} if it
can be realized as the intersection graph of a low-density collection
of objects in low dimensions.

There is much work on intersection graphs, from interval graphs, to
unit disk graphs, and more. The circle packing theorem
\cite{k-kdka-36,a-ocpls-70,pa-cg-95} implies that every planar graph
can be realized as a coin graph, where the vertices are interior
disjoint disks, and there is an edge connecting two vertices if their
corresponding disks are touching. This implies that planar graphs are
low density.  Miller \etal \cite{mttv-sspnng-97} studied the
intersection graphs of balls (or fat convex object) of bounded depth
(i.e., every point is covered by a constant number of balls), and
these intersection graphs are readily low density. Some results
related to our work include: (i) planar graphs are the intersection
graph of segments \cite{cg-epgig-09}, and (ii) string graphs (i.e.,
intersection graph of curves in the plane) have small separators
\cite{m-nossg-14}.

\paragraph{Polynomial expansion.}

The class of low-density graphs is contained in the class of graphs
with polynomial expansion. This class was defined by \Nesetril and
Ossona \si{de} Mendez as part of a greater investigation on the
sparsity of graphs (see the book \cite{no-s-12}). A motivating
observation to their theory is that the sparsity of a graph (the ratio
of edges to vertices) is not necessarily sufficient for tractability.
For example, a clique (which is a graph with maximum density) can be
disguised as a sparse graph by splitting every edge by a middle
vertex. Furthermore, constant degree expanders are also sparse.  For
both graphs, many optimization problems are intractable (intuitively,
because they do not have a small separator).

Graphs with bounded expansion are \emph{nowhere dense graphs}
\cite[Section 5.4]{no-s-12}. Grohe \etal \cite{gks-dfopndg-14}
recently showed that first-order properties are fixed-parameter
tractable for nowhere dense graphs. In this paper, we study graphs of
polynomial expansion \cite[Section 5.5]{no-s-12}, which intuitively
requires a graph to not only be sparse, but have shallow minors that
are sparse as well.

It is known that graphs with polynomial expansion have sublinear
separators \cite{no-gcbe1-08}. The converse, that any graph that has
hereditary sublinear separators has polynomial expansion, was recently
shown by Dvo{\v{r}}{\'{a}}k and Norin \cite{dn-ssspe-15}. As such, our
work looks beyond the geometric setting to consider the general role
of separators in approximation.

\subsubsection{Further related work}

There is a long history of optimization in structured graph
classes. Lipton and Tarjan first obtained a \PTAS for independent set
in planar graphs by using separators
\cite{lt-stpg-79,lt-apst-80}. Baker \cite{b-aancp-94} developed
techniques for covering problems (e.g.\ dominated set) on planar
graphs. Baker's approach was extended by Eppstein \cite{e-dtmcg-00} to
graphs with bounded local treewidth, and by Grohe \cite{g-ltwem-03} to
graphs excluding minors. Separators have also played a key role in
geometric optimization algorithms, including:
\begin{inparaenum}[(i)]
\item \PTAS for independent set and (continuous) piercing set for fat
  objects \cite{c-ptasp-03, mr-irghs-10},
  %
\item \QPTAS for maximum weighted independent sets of polygons
  \cite{aw-asmwi-13,aw-qmwis-14,h-qssp-14}, and
\item \QPTAS for \ProblemC{Set Cover} by pseudodisks
  \cite{mrr-qgscp-14}, among others.
\end{inparaenum}
Lastly, Cabello and Gajser \cite{cg-spfgem-14} develop \PTAS{}'s for
some of the problems we study in the specific setting of minor-free
graphs.

\begin{figure}[t]
  \begin{center}
    \begin{tabular}{|l|l|l|}
      \hline
      Objects & \si{Approx.} \si{Alg.} & Hardness\\
      \hline%
      \ProblemE{Disks/pseudo-disks}%
      {\PTAS \cite{mr-irghs-10}}%
      {\smallskip%
      Exact version \NPHard\\%
      via point-disk duality %
      \cite{fg-oafac-88}%
      }
      \hline
      \ProblemE{Fat triangles of similar size.}
      {$O( \log \log \mathrm{opt} )\bigl.$ %
      \cite{aes-ssena-10}%
      }{\smallskip%
      \APXHard: \lemref{no:PTAS:fat:hit:set}}
      \hline
      \ProblemE{Objects with $O(1)$ density.}
      {\PTAS:     \thmref{g:hitting:set:cover}$\bigl.$}%
      {Exact \si{ver.} \NPHard \cite{fg-oafac-88}}
      \hline
      \ProblemE{Objects $\polylog$ density.}
      {\QPTAS:     \thmref{g:hitting:set:cover}}%
      {\smallskip%
      No \PTAS under \ETH\\
      \lemref{e:t:h:g:polylog:d} /
      \lemrefshort{no:PTAS:fat:hit:set}}
      \hline%
      \ProblemE{Objects with  density $\cDensity$ in $\Re^d$}%
      {%
      \PTAS:
      \thmref{g:hitting:set:cover}%
      \\ run time
      $n^{O(\cDensity^{(d+1)/d}/\epsilon^d)}$
      }{%
      No $(1+\eps)$-approx

      with RT
      $n^{\poly( \log \cDensity, 1/\eps)}$

      assuming \ETH: %
      \lemrefshort{e:t:h:g:polylog:d}%
      }
      \hline
    \end{tabular}
  \end{center}
  \vspace{-0.3cm}
  \caption{Known results about the complexity of \emph{discrete}
    geometric hitting set. The input is a set of points, and a set of
    objects, and the task is to find the smallest subset of points
    such that any object is hit by one of these points.  }
  \figlab{hitting:set:summary}
\end{figure}

\subsection{Our results}

We systematically study the class of graphs that have low density,
first proving that they have polynomial expansion.  We then develop
approximation algorithms for this broader class of graphs, as follows:

\smallskip%
\SaveIndent%
\begin{compactenum}[(A)]%
  \RestoreIndent%
  \setlength{\itemsep}{3pt}
\item \textbf{\PTAS for independent set.} %
  For graphs that have sublinear hereditary separators we show \PTAS
  for independent set, see \secref{approx:v:separators}. This covers
  graphs with low density and polynomial expansion. These results are
  not surprising in light of known results \cite{ch-aamis-12}, but
  provide a starting point and contrast for subsequent results.

\item \textbf{\PTAS for packing problems.}  The above \PTAS also hold
  for packing problems, such as finding maximal induced planar
  subgraph, and similar problems, see \exmref{geometric:packing} and
  \lemref{indep:easy}.

\item \textbf{\PTAS for independent/packing when the output is
    sparse.} %
  More surprisingly, one get a \PTAS even if the subgraph induced on
  the union of two solutions has polynomial expansion. Thus, while the
  input may not be sparse, as long as the output is sparse, one can
  get an efficient approximation algorithms, see \thmref{independent}.

  In particular, this holds if the output is required to have low
  density, because the union of two sets of objects with low density
  is still low density.  The resulting algorithms in the geometric
  setting are faster than those for polynomial expansion graphs, by
  using the underlying geometry of low-density graphs.

\item \textbf{\PTAS for dominating set.}~%
  Low density graphs remain low density even if one merges locally
  objects that are close together, see
  \lemref{density:shallow:minors}. More generally, if one considers a
  collection of $t$-shallow subgraphs (i.e., subgraphs with radius $t$
  in the edge distance) of a polynomial expansion graph, then their
  intersection graph also has polynomial expansion, as long as
  \emph{no} vertex in the original graph participates in more than
  constant number of subgraphs.

  This surprising property implies that local search algorithms
  provides a \PTAS for problems like \ProblemC{Dominating Set} for
  graphs with polynomial expansion, see \secref{dom:set}.

\item \textbf{\PTAS for multi-cover dominating set with reach
    constraints.}~%
  These results can be extended to multi-cover variants of dominating
  set for such graphs, where every vertex can be asked to be dominated
  a certain number of times, and require that the these dominated
  vertices are within a certain distance.  See
  \lemref{ptas:subset:dom:2}.

\item \textbf{Connected dominating set.}~%
  The above algorithms also extend to a \PTAS for connected dominating
  set, see \secref{connected:dominating:set}.

\item \textbf{\PTAS for vertex cover for graphs with polynomial
    expansion.} See \obsref{v:c:poly:expansion}.

\item \textbf{\PTAS for geometric hitting set and set cover.}~%
  The new algorithms for dominating sets readily provides \PTAS's for
  discrete geometric set cover and hitting set for low density inputs,
  see \secref{geometric:applications}.

\item \textbf{Hardness of approximation.}~The low-density algorithms
  are complimented by matching hardness results that suggest the new
  approximation algorithms are nearly optimal with respect to depth
  (under \SETH: the assumption that \ProblemC{SAT} over $n$ variables
  can not be solved in better than $2^n$ time).
\end{compactenum}
\smallskip%
The context of our results, for geometric settings, is summarized in
\figref{set:cover:summary} and \figref{hitting:set:summary}.

\paragraph{Sparsity is not enough.}

It is natural to hope that the above algorithms would work for sparse
graphs (i.e., graphs that have linear number of edges). Unfortunately,
as mentioned earlier, constant degree expanders, which play a
prominent rule in theoretical computer science, are sparse, and the
above algorithms fail for them as they do not have separators.

\paragraph{Low level technical contributions.}

We show that graphs with polynomial expansion retain polynomial
expansion even if one is allowed to locally connect a vertex to other
nearby vertices in a controlled way. To this end, we extend the notion
of $t$-shallow minors to shallow packings (see
\defref{shallow:packing}). We then use a probabilistic argument to
show that the associated intersection graph still has polynomial
expansion, see \lemref{edge-density-shallow-cover}. The proof of this
lemma is elegant and might be of independent interest.

\bigskip
\noindent%
\textbf{Paper organization.} %
We describe low-density graphs in \secref{low:density} and prove some
basic properties.  Bounded expansion graphs are surveyed in
\secref{poly-expansion}.  \secref{approx:algorithms} present the new
approximation algorithms.  \secref{hardness} present the hardness
results. Conclusions are provided in \secref{conclusions}.

\section{Preliminaries} %

\subsection{Low-density graphs} %
\seclab{low:density}

\begin{defn}
  For a graph $\graph = (\Vertices,\Edges)$, and any subset
  $\SetA \subseteq \Vertices$, let $\GInduced{\SetA}$ denote the
  \emphi{induced subgraph} of $\graph$ over $\SetA$. Formally, we have
  \begin{math}
    \GInduced{\SetA} = \pth{ \SetA, \Set{uv}{u,v \in \SetA, \text{ and
        } uv \in \Edges \bigr.}\bigr.}.
  \end{math}
\end{defn}

\begin{defn}
  Consider a set of objects $\ObjSet$.  The \emphi{intersection graph}
  of $\ObjSet$, denoted by $\IGraph{\ObjSet}$, is the graph having
  $\ObjSet$ as its set of vertices, and an edge between two objects
  $\obj,\objA \in \ObjSet$ if they intersect.  Formally,
  \begin{math}
    \IGraph{\ObjSet} = \pth{\bigl.\ObjSet, %
      \Set{\bigl.\obj \objA}{\obj, \objA \in \ObjSet %
        \text{ and } \obj \cap \objA \neq \emptyset}}.
  \end{math}
\end{defn}

One of the two main thrusts of this work is investigating the
following family of graphs.

\begin{defn}%
  \deflab{low:density}%
  A set of objects $\ObjSet$ in $\Re^d$ (not necessarily convex or
  connected) has \emphi{density $\cDensity$} if any object $\obj$ (not
  necessarily in $\ObjSet$) intersects at most $\cDensity$ objects in
  $\ObjSet$ with diameter greater than or equal to the diameter of
  $\obj$. The minimum such quantity is denoted by
  $\densityX{\ObjSet}$.  If $\cDensity$ is a constant, then $\ObjSet$
  has \emphi{low density}.

  A graph that can be realized as the intersection graph of a set of
  objects $\ObjSet$ in $\Re^d$ with density $\cDensity$ is
  \emphi{$\cDensity$-dense}.
\end{defn}%

\begin{defn}%
  \deflab{degenerate}%
  A graph $\graph$ is \emphi{$k$-degenerate} if any subgraph of
  $\graph$ has a vertex of degree at most $k$.
\end{defn}

\begin{observation}
  \obslab{low:density:g:sparse}%
  A $\cDensity$-dense graph $\graph$ is
  $(\cDensity-1)$-degenerate. Indeed, consider the set of objects
  $\ObjSet$ that induces $\graph$. Let $\obj$ be the object with
  smallest diameter $d_0$ in $\ObjSet$. By choice of $\obj$, any other
  object intersecting $\obj$ has diameter at least $d_0$. Since at
  most $\cDensity$ objects, in $\ObjSet$, with diameter at least $d_0$
  intersect $\obj$ (including $\obj$ itself), the degree of $\obj$ in
  $\graph$ is $\cDensity-1$. Clearly, this argument applies to any
  subgraph of $\graph$.
\end{observation}%

\subsubsection{Fatness and density}
For $\alpha > 0$, an object $\objA \subseteq \Re^d$ is
\emphi{$\alpha$-fat} if for any ball $\ball$ with a center inside
$\objA$, that does not contain $\objA$ in its interior, we have
$\volX{\ball \cap \objA} \geq \alpha \volX{\ball}$
\cite{bksv-rimga-02}%
\footnote{There are several different, but roughly equivalent,
  definitions of fatness in the literature, see \si{de} Berg
  \cite{b-ibucf-08} and the followup work by Aronov \etal
  \cite{abes-ibulf-14} for some recent results. In particular, our
  definition here is what \si{de} Berg refers to as being
  \emph{locally fat}.}. A set $\Family$ of objects is \emphi{fat} if
all its members are $\alpha$-fat for some constant $\alpha$.  A
collection of objects $\ObjSet$ has \emphi{depth} $k$ if any point in
the underlying space lies in at most $k$ objects of $\ObjSet$. The
depth index of a set of objects is a lower bound on the density of the
set, as a point can be viewed as a ball of radius zero. The following
is well known, and we include a proof for the sake of completeness.

\begin{lemma}
  A set $\Family$ of $\alpha$-fat objects in $\Re^d$ with depth $k$
  has density $\nfrac{k 2^d}{\alpha}$. In particular, if $\alpha, k $
  and $d$ are bounded constants, then $\Family$ has low density.
\end{lemma}

\begin{wrapfigure}{r}{.16\textwidth}
  \vspace{-2em}
  \hfill\includegraphics[width=.16\textwidth]{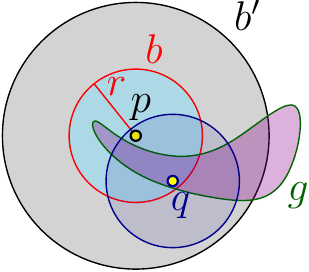}
  \vspace{-2em}
\end{wrapfigure}
\noindent\textit{Proof:}
Let $\ball = \ballY{\pnt}{r}$ be any ball in $\Re^d$.  Let $\FamilyA$
be the set of all objects in $\Family$ that have diameter larger than
$\ball$ and intersect it, and consider any object
$\objA \in \FamilyA$.  A ball $\ballY{\pntA}{r}$ centered at a point
$\pntA \in \objA \cap \ball$ does not contain $\objA$ in its interior
because $\diamX{\ballY{\pntA}{r}} \leq \diamX{\objA}$.  By the
definition of $\alpha$-fatness, we have
\begin{align*}
  \volX{\bigl.\ballY{\pntA}{r} \cap \objA}%
  \geq%
  \alpha \volX{\bigl.\ballY{\pntA}{r}} = \alpha \volX{\ball}.
\end{align*}
Furthermore, $\ballY{\pntA}{r}$ is contained in the ball
$\ballA = \ballY{\pnt}{2r}$, and
\begin{align*}
  \volX{\bigl.\ballA \cap \objA} %
  \geq %
  \volX{\bigl.\ballY{\pntA}{r} \cap \objA} %
  \geq %
  \alpha \volX{\ball} %
  = %
  \frac{\alpha}{2^d} \volX{\ballA}.
\end{align*}
By assumption, each point in $\ballA$ can be covered by at most $k$
objects of $\FamilyA$, hence
\begin{align*}
  k \volX{\ballA}%
  \geq %
  \sum_{\objA \in \FamilyA} \volX{\bigl.\ballA \cap \objA} %
  \geq%
  \cardin{\FamilyA}
  \frac{\alpha}{2^d} \volX{\ballA}.
\end{align*}
Thus, $\cardin{\FamilyA} \leq k(2^d/\alpha)$, bounding the number of
``large'' objects of $\ObjSet$ intersecting $\ball$.%
\hfill\myqedsymbol

\begin{defn}
  A metric space $\Metric$ is a \emph{doubling space} if there is a
  universal constant $\dblC > 0$ such that any ball $\ball$ of radius
  $r$ can be covered by $\dblC$ balls of half the radius. Here $\dblC$
  is the \emphi{doubling constant}, and its logarithm is the
  \emph{doubling dimension}.
\end{defn}

\begin{observation}
  \obslab{d:c:d}%
  In $\Re^d$ the doubling constant is $\dblCd = 2^{O(d)}$, and the
  doubling dimension is $O(d)$ \cite{v-cbseb-05}, making the doubling
  dimension a natural abstraction of the notion of dimension in the
  Euclidean case.
\end{observation}

\begin{lemma}
  \lemlab{density:smaller:objects}%
  Let $\ObjSet$ be a set of objects in $\Re^d$ with density
  $\cDensity$. Then, for any $\alpha \in (0,1)$, a ball
  $\ball = \ballY{\cen}{r}$ can intersect at most
  $\cDensity\dblCd^{\ceil{\lg 1/\alpha}}$ objects of $\ObjSet$ with
  diameter $\geq 2r \alpha$, where $\lg$ is the $\log$ function in
  base two, and $\dblCd$ is the doubling constant of $\Re^d$.
\end{lemma}
\begin{proof}
  By the definition of the doubling constant, one can cover $\ball$ by
  $\pth{\dblCd}^i$ balls of radius $r/2^i$.  As such, one can cover
  $\ball$ with $\leq \dblCd^{\ceil{\log_2 1/\alpha}}$ balls of radius
  $\leq \alpha r$. Each of these balls, by definition of density, can
  intersect at most $\cDensity$ objects of $\ObjSet$ of diameter at
  least $2r \alpha $.
\end{proof}

The density definition can be made to be somewhat more flexible, as
follows.

\begin{lemma}
  Let $\beta > 1$ be a parameter, and let $\ObjSet$ be a collection of
  objects in $\Re^d$ such that, for any $r$, any ball with radius $r$
  intersects at most $\cDensity$ objects with diameter
  $\geq 2 r \beta$. Then $\ObjSet$ has density
  $\dblCd^{\ceil{ \lg \beta }} \cDensity$.
\end{lemma}
\begin{proof}
  Let $\ball$ be a ball with radius $r$. We can cover $\ball$ with
  $\dblCd^{\ceil{ \lg \beta }}$ balls with radius $r/\beta$. Each
  $(r/\beta)$-radius ball can intersect at most $\cDensity$ objects
  with diameter larger than $2(r/\beta) \beta = 2r$, so $\ball$
  intersects at most $\dblCd^{\ceil{ \lg \beta }} \cDensity$ objects
  with diameter at least $2r = \diamX{\ball}$.
\end{proof}

\subsubsection{Minors of objects}

\begin{defn}%
  \deflab{t:shallow}%
  A graph $\graph$ is \emphi{$t$-shallow} if there is a vertex
  $h \in \VerticesX{\graph}$, such that for any vertex
  $u \in \VerticesX{\graph}$ there is a path $\pi$ that connects $h$
  to $u$, and $\pi$ has at most $t$ edges. The vertex $h$ is a
  \emphi{center} of $\graph$, denoted by $h =\centerX{\graph}$.  The
  minimum integer $t$ such that $\graph$ is $t$-shallow is the
  \emphi{radius} of $\graph$.
\end{defn}

Let $\ObjSet$ and $\ObjSetA$ be two sets of objects in $\Re^d$. The
set $\ObjSetA$ is a \emphi{minor} of $\ObjSet$ if it can be obtained
by deleting objects and replacing pairs of overlapping objects $\obj$
and $\objA$ (i.e., $\obj \cap \objA \neq \emptyset$) with their union
$\obj \cup \objA$.  Consider a sequence of unions and deletions
operations transforming $\ObjSet$ into $\ObjSetA$.  Every object
$\objA \in \ObjSetA$ corresponds to a set of objects of
$\clusterX{\objA} \subseteq \ObjSet$, such that
$\cup_{\objB \in \clusterX{\objA}} \objB = \objA$.  The set
$\clusterX{\objA}$ is a \emphi{cluster} of objects of $\ObjSet$.

Surprisingly, even for a set $\Family$ of fat and convex shapes in the
plane with constant density, their intersection graph
$\IGraph{\Family}$ can have arbitrarily large cliques as minors (see
\figref{clique:minor}). Note that the clusters in
\figref{clique:minor} induce intersection graphs with large graph
radius.

\begin{defn}
  For sets of objects $\ObjSet$ and $\ObjSetA$, if %
  \smallskip%
  \begin{compactenum}[\quad(i)]
  \item $\ObjSetA$ is a minor of $\ObjSet$, and
  \item the intersection graph of each cluster of $\ObjSet$ (that
    corresponds to an object in $\ObjSetA$) is $t$-shallow,
  \end{compactenum}%
  \smallskip%
  then $\ObjSetA$ is a \emphi{$t$-shallow minor} of $\ObjSet$.
\end{defn}

The following lemma shows that there is a simple relationship between
the depth of a shallow minor of objects and its density.

\begin{figure}[t]
  \centerline{%
    \begin{tabular}{c%
      c%
      c%
      c%
      c}
      \IncludeGraphics[page=1,
      width=.17\textwidth]{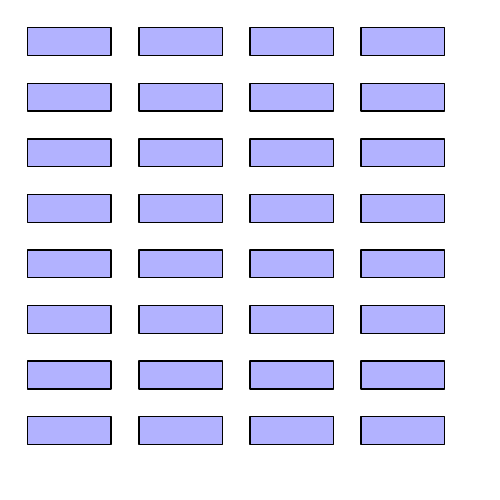}%
      &
        \IncludeGraphics[page=2,
        width=.17\textwidth]{figs/clique_minor}%
      &
        \IncludeGraphics[page=3,
        width=.17\textwidth]{figs/clique_minor}%
      &
        \IncludeGraphics[page=4,
        width=.17\textwidth]{figs/clique_minor}%
      &
        \IncludeGraphics[page=5,
        width=.17\textwidth]{figs/clique_minor}%
      \\
      (A)%
      &%
        (B)%
      &%
        (C)%
      &%
        (D)%
      &%
        (E)
    \end{tabular}%
  }

  \caption{%
    \small %
    (A) and (B) are two low-density collections of $n^2$ disjoint
    horizontal slabs, whose intersection graph (C) contains $n$ rows
    as minors. (D) is the intersection graph of a low-density
    collection of vertical slabs that contain $n$ columns as
    minors. In (E), the intersection graph of all the slabs contain
    the $n$ rows and $n$ columns as minors that form a $K_{n,n}$
    bipartite graph, which in turn contains the clique $K_{n}$ as a
    minor.  }
  \figlab{clique:minor}%
  \vspace{-1em}
\end{figure}%

\begin{lemma}%
  \lemlab{density:shallow:minors}%
  Let $\ObjSet$ be a collection of objects with density $\cDensity$ in
  $\Re^d$, and let $\ObjSetA$ be a $t$-shallow minor of
  $\ObjSet$. Then $\ObjSetA$ has density at most
  $ t^{O(d)} \cDensity$.
\end{lemma}
\begin{proof}
  Every object $\objA \in \ObjSetA$ has an associated cluster
  $\clusterX{\objA} \subseteq \ObjSet$ such that
  $\objA = \bigcup_{\obj \in \clusterX{\objA}} \obj$, and these
  clusters are disjoint.  Let
  $\Partition = \Set{\clusterX{\objA}}{\objA \in \ObjSetA}$ be the
  induced partition of $\ObjSet$ into clusters (which may be a
  partition of a subset of $\ObjSet$). Consider any ball
  $\ball = \ballY{\cen}{r}$, and suppose that $\objA \in \ObjSetA$
  intersects $\ball$ and has diameter at least $2r$. Let
  $\clusterX{\objA} \in \Partition$ be its cluster, and let
  $\graphA = \IGraph{\clusterX{\objA}}$ be its associated intersection
  graph.  By assumption, $\graphA$ has (graph) radius at most $t$, and
  diameter at most $2t$.

  Let $\objB$ be any object in $\clusterX{\objA}$ that intersect
  $\ball$, let $\distCharX{\graphA}$ denote the shortest path metric
  of $\graphA$ (under the number of edges), and let $\objB'$ be the
  object in $\clusterX{\objA}$ closest to $\objB$ (under
  $\distCharX{\graphA}$), such that $\diamX{\objB'} \geq r/t$. If
  there is no such object then the diameter of
  $\diamX{\objA} < 2t(r/t) \leq 2r$, which is a contradiction.

  Consider the shortest path $\pi \equiv \objB_1,\ldots, \objB_m$
  between $\objB=\objB_1$ and $\objB'=\objB_m$ in $\graphA$, where
  $m \leq 2t$.  By the choice of $h$, we have $\diamX{\objB_i} < r/t$,
  for $i=1, \ldots m-1$, and the distance between $\ball$ and $\objB'$
  is bounded by
  \begin{math}
    \sum_{i=1}^{m-1} \diamX{\objB_i} \leq (m-1) (r/t) < 2r.
  \end{math}
  The object $\objB'$ is the \emph{representative} of $\objA$, denoted
  by $\repX{\objA} \in \clusterX{\objA}$.

  Now, let
  \begin{math}
    \ObjSetB%
    =%
    \Set{\bigl.\repX{\objA} }{\objA \in \ObjSetA, \diamX{\objA} \geq
      2r, \text{ and } \objA \cap \ball \neq \emptyset} \subseteq
    \ObjSet
  \end{math}
  be the representatives of the large objects in $\ObjSetA$
  intersecting $\ball$.  The representatives in $\ObjSetB$ are all
  distinct, have diameters $\geq r/t$, intersect $\ballY{\cen}{3r}$,
  and belong to $\ObjSet$ - a set with density $\cDensity$.  Setting
  $\alpha = 1/6t$, \lemref{density:smaller:objects} implies that
  $\cardin{\ObjSetB} \leq \cDensity\dblCd^{\ceil{\lg (6t)}}$.  Since
  $\dblCd = 2^{O(d)}$ \cite{v-cbseb-05}, it follows that
  $\cardin{\ObjSetB} = t^{O(d)}$, implying the claim.
\end{proof}

\subsection{Graphs with polynomial expansion}
\seclab{poly-expansion} %

\subsubsection{Basic properties}

\begin{defn}
  \deflab{shallow:minor}%
  Let $\graph$ be an undirected graph. A \emphi{minor} of $\graph$ is
  a graph $\graphA$ that can be obtained by contracting edges,
  deleting edges, and deleting vertices from $\graph$. If $\graphA$ is
  a minor of $\graph$, then each vertex $v$ of $\graphA$ corresponds
  to a \emphi{cluster} -- a connected set $\clusterX{v}$ of vertices
  in $\graph$ (i.e., these are the vertices of a tree in the forest
  formed by the contracted edges).  The graph $\graphA$ is a
  \emphi{$t$-shallow minor} (or a \emphi{minor of depth $t$}) of
  $\graphA$, where $t$ is an integer, if for each vertex
  $v \in \VerticesX{\graphA}$, the induced subgraph
  $\IGraph{\cluster}$ of the corresponding cluster
  $\cluster = \clusterX{v}$ is $t$-shallow (see \defref{t:shallow}).
  Let $\minorsDY{t}{\graphA}$ denote the set of all graphs that are
  minors\footnote{I.e., these graphs can not legally drink alcohol.}
  of $\graphA$ of depth $t$.
\end{defn}

\begin{defn}[{\cite{no-gcbe1-08}}]
  \deflab{r:shallow:density}%
  The \emph{greatest reduced average density of rank $r$}, or just the
  \emphi{$r$-shallow density}, of $\graph$ is the quantity
  \begin{math}
    \displaystyle%
    \gradY{r}{\graph} %
    = %
    \sup_{\graphA \in \minorsDY{r}{\graph}}
    \frac{\cardin{\EdgesX{\graphA}}}{\cardin{\verticesof{\graphA}}}.%
  \end{math} %
\end{defn}

\begin{defn}
  \deflab{expansion}%
  A graph \emphi{class} is a (potentially infinite) set of graphs
  (e.g., the class of planar graphs).  The \emphi{expansion} of a
  graph class $\class$ is the function
  \begin{math}
    f: \naturalnumbers \to \naturalnumbers \cup \setof{\infty}
  \end{math}
  defined
  by
  \begin{math}
    f(r) = \sup_{\graph \in \class} \gradY{r}{\graph}. %
  \end{math}
  The class $\class$ has \emphi{bounded expansion} if $f(r)$ is finite
  for all $r$. Specifically, a class $\class$ with bounded expansion
  has \emphi{polynomial expansion} (resp., \emph{subexponential
    expansion} or \emph{constant expansion}) if $f$ is bounded by a
  polynomial (resp., subexponential function or constant).  The
  polynomial expansion is of \emphi{order $k$} if $f(x) = O(x^k)$.
  Naturally, a graph $\graph$ has \emph{polynomial expansion} of order
  $k$ if it belongs to a class of graphs with {polynomial expansion}
  of order $k$.
\end{defn}%

\begin{observation}
  \obslab{bounded:expansion:degenerate}%
  If a graph $\graph$ has bounded expansion, then $\graph$ has average
  degree at most
  \begin{math}
    \mu = \gradY{1}{\graph} / 2 = O(1),
  \end{math}
  by taking the graph $\graph$ as its own $1$-shallow minor (with
  every vertex is its own cluster). In particular, the vertex $v_0$
  with minimum degree has degree at most $\mu$. Similarly, any
  subgraph of $\graph$ has a vertex $v_1$ with degree at most $\mu$,
  so the graph $\graph$, by virtue of its bounded expansion, is
  $O(1)$-degenerate (see \defref{degenerate}).
\end{observation}

As an example of a class of graphs with constant expansion, observe
that planar graphs have constant expansion because a minor of a planar
graph is planar, and by Euler's formula, every planar graph is
sparse. More surprisingly, \lemref{density:shallow:minors} together
with \obsref{low:density:g:sparse} implies that low-density graphs
have polynomial expansion.

\begin{lemma}
  Let $\cDensity > 0$ be fixed.  The class of $\cDensity$-dense graphs
  in $\Re^d$ has polynomial expansion bounded by
  $f(r) = \cDensity r^{O(d)}$.
\end{lemma}

\subsubsection{Separators}

\begin{defn}%
  \deflab{separator}%
  Let $\graph = (\Vertices,\Edges)$ be an undirected graph.  Two sets
  $\SetA, \SetB \subseteq \Vertices$ are \emphi{separate} in $\graph$
  if
  \begin{inparaenum}[(i)]
  \item $\SetA$ and $\SetB$ are disjoint, and
  \item there is no edge between the vertices of $\SetA$ and $\SetB$
    in $\graph$.
  \end{inparaenum}
  A set $\SepSet \subseteq \Vertices$ is a \emphi{separator} for a set
  $\SetC \subseteq \Vertices$, if
  $\cardin{\SepSet} = o\pth{\cardin{\SetC}}$, and
  $\SetC \setminus \SepSet$ can be partitioned into two
  \emph{separate} sets $\SetA$ and $\SetB$, with
  \begin{math}
    \cardin{\SetA} \leq (2/3)\cardin{\SetC}
  \end{math}
  and
  \begin{math}
    \cardin{\SetB} \leq (2/3)\cardin{\SetC}.
  \end{math}
  (Here the choice of $2/3$ is arbitrary, and any constant smaller
  than $1$ is sufficient.)
\end{defn}

\Nesetril and Ossona \si{de} Mendez showed that graphs with
subexponential expansion have subexponential-sized separators. For the
simpler case of polynomial expansion, we have the following\footnote{A
  proof is also provided in \cite{hq-naape-16-arxiv}.}.

\begin{theorem}%
  [\expandafter{\cite[Theorem 8.3]{no-gcbe2-08}}]%
  \thmlab{p:e:separator}%
  Let $\class$ be a class of graphs with polynomial expansion of order
  $k$. For any graph $\graph \in \class$ with $n$ vertices and $m$
  edges, one can compute, in $O\pth{m n^{1-\alpha} \log^{1-\alpha} n}$
  time, a separator of size
  \begin{math}
    O \bigl( n^{1 - \alpha} \log^{1 - \alpha} n \bigr),
  \end{math}
  where $\alpha = 1/\pth{2k+2}$.%
\end{theorem}

\thmref{p:e:separator} yields a sublinear separator for low-density
graphs of size
\begin{math}
  O\Bigl( %
  \pth{ %
    \rho^2 n \log n %
  }^{1 - \frac{1}{O\pth{\log \dblC} } } %
  \Bigr).
\end{math}
Geometric arguments give a somewhat stronger separator. For the sake
of completeness, we provide next a proof of the following result, but
we emphasize that it is essentially already known
\cite{mttv-sspnng-97,sw-gsta-98,c-ptasp-03}. This proof is arguably
simpler and more elegant than previous proofs.

\begin{figure}
  \noindent%
  \begin{minipage}[b]{0.29\linewidth}
    (A) \!\!\!\!\hspace{-0.3cm}%
    \IncludeGraphics[page=1]{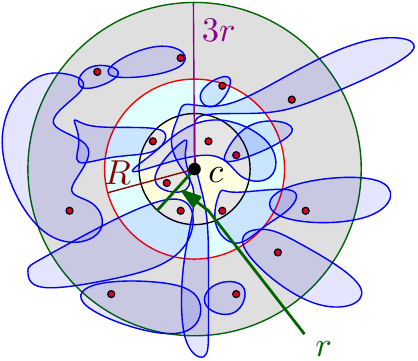}%
    \hspace*{-3.11cm}%
  \end{minipage}%
  \quad%
  \begin{minipage}[b]{0.33\linewidth}
    (B)\!\!\!\!%
    \IncludeGraphics[page=2]{figs/separator}
    \hspace*{-3.11cm}%
  \end{minipage}%
  \qquad\begin{minipage}[b]{0.3\linewidth}
    \captionof{figure}{%
      Illustration of the proof of
      \lemref{separator:low:density}.\\[0.2cm]
      (A) %
      The ball $\ballY{\cen}{r}$, and the separating sphere
      $\sphereX{\cen}{R}$.  %
      \\[0.2cm]%
      (B) %
      All the objects intersecting $\sphereX{\cen}{R}$ are in the
      separating set.}
  \end{minipage}
\end{figure}
\begin{lemma}
  \lemlab{separator:low:density}%
  Let $\ObjSet$ be a set of $n$ objects in $\Re^d$ with density
  $\cDensity > 0$ (see \defref{low:density}), and let $k \leq n$ be
  some prespecified number.  Then, one can compute, in expected $O(n)$
  time, a sphere $\sphereC$ that intersects
  $O\pth{\cDensity + \cDensity^{1/d} k^{1-1/d} }$ objects of
  $\ObjSet$. Furthermore, the number of objects of $\ObjSet$ strictly
  inside $\sphereC$ is at least $k - o(k)$, and at most $O(k)$.  For
  $k = O(n)$ this results in a balanced (global) separator. %
  (Note that the $O$ notation hides constants that depend on $d$.)
\end{lemma}

\begin{proof}
  For every object $\obj \in \ObjSet$, choose an arbitrary
  representative point $\pnt_{\obj} \in \obj$. Let $\PntSet$ be the
  resulting set of points. Let $\ballY{\cen}{r}$ be the smallest ball
  containing $k$ points of $\PntSet$. As in \cite{h-speps-13},
  randomly pick $R$ uniformly in the range $[r, 2r]$. We claim that
  the sphere $\sphereC = \sphereX{\cen}{R}$ bounding the ball
  $\ball = \ballY{\cen}{R}$ is the desired separator.

  To this end, consider the distance $\ell = t\cdot r$, where
  $t \in (0,1)$ is some real number to be specified shortly.  We count
  the number of objects intersecting $\sphereC$ as follows:

  \SaveIndent%
  \medskip%
  \begin{compactenum}[(A)]%
    \RestoreIndent%
  \item \textbf{Large objects (diameter $\geq \ell$).} %
    The sphere $\sphereC$ can intersect only
    $N_1 = O( \cDensity + \cDensity/t^{d-1})$ objects with diameter
    $\geq \ell$. Indeed, cover the sphere $\sphereC$ with
    $O( 1/t^{d-1})$ balls of radius $\ell/2$, and let $\BallSet$ be
    this set of balls. Next, charge each object of diameter larger
    than $\ell$ intersecting $\sphereC$ to the ball of $\BallSet$ that
    intersects it.  Each ball of $\BallSet$ get charged at most
    $\cDensity$ times.
    \smallskip%
  \item \textbf{Small objects (diameter $< \ell$).} %
    Let $\ObjSet_{\ballA}$ be the set of objects of $\ObjSet$ with
    diameter $\leq \ell$ fully contained in
    $\ballA = \ballY{\cen}{3r}$, which contains all the small objects
    of $\ObjSet$ that intersects $\sphereC$. The ball $\ballA$ can be
    covered by $\dblCd^2$ balls of radius $r$, where $\dblCd$ is the
    doubling constant of $\Re^d$ (see \obsref{d:c:d}), and each ball
    of radius $r$ contains at most $k$ representative points. Thus,
    $\ballA$ contains at most $k' = \dblCd^2 k$ points of $\PntSet$,
    implying that $\cardin{\ObjSet_{\ballA}} \leq k' = O\pth{k}$.

    For an object $\objA \in \ObjSet_{\ballA}$, consider the closest
    point $\pnt$ and the furthest point $\pntA$ in $\objA$ from
    $\cen$.  The object $\objA$ is in the separating set (and
    ``intersects'' $\sphereC$) if $\sphereC$ separates $\pnt$ from
    $\pntA$ ($\objA$ may be disconnected).  As $R$ is chosen uniformly
    at random from $[r, 2r]$, we have that
    \begin{align*}
      \alpha(\objA)%
      =%
      \Prob{\Bigl.\objA \text{ intersects }\sphereC}%
      \leq%
      \frac{\distY{\cen}{\pntA} - \distY{\cen}{\pnt}}{r}%
      \leq%
      \frac{\diamX{\objA}}{r} %
      \leq%
      \frac{\ell}{r}%
      =%
      t.
    \end{align*}
    The expected number of objects of $\ObjSet_{\ballA}$ intersecting
    $\sphereC$ is
    \begin{math}
      N_2 =%
      \sum\nolimits_{\objA \in \ObjSet_{\ballA}} \alpha(\objA)%
      \leq%
      t k'%
      =%
      O( t k).
    \end{math}
  \end{compactenum}
  \smallskip%
  We conclude that the separator size, in expectation, is
  \begin{math}
    N%
    = %
    N_1 + N_2 =%
    O\pth{ \Bigl.  \cDensity + \cDensity/t^{d-1} + k t }.
  \end{math}
  Solving for $ \cDensity/ t^{d-1} = k t$ yields
  $t = (\cDensity/k)^{1/d}$, and the resulting separator is (in
  expectation) of size
  $N = O(\cDensity + \cDensity^{1/d} k^{1-1/(d)} )$.

  As for the running time, it is sufficient to find a two
  approximation to the smallest ball that contains $k$ points of
  $\PntSet$, and this can be done in linear time
  \cite{hr-nplta-13}. Using such an approximation slightly
  deteriorates the constants in the bounds. By Markov's inequality,
  $\sphereC$ intersects at most $2N$ objects of $\ObjSet$ with
  probability $\geq 1/2$. If this is not true, we rerun the
  algorithm. In expectation, the algorithm succeeds in finding a
  sphere that intersects at most $2N$ objects of $\ObjSet$ within a
  constant number of iterations.
\end{proof}

\begin{remark}
  \obslab{improved:bound}%
  \si{Mark de Berg} (personal communication) pointed out the current
  simplified proof of \lemref{separator:low:density}. The authors
  thank him for pointing out the simpler proof.
\end{remark}

It was recently shown that any graph with strongly sublinear
hereditary separators has polynomial expansion \cite{dn-ssspe-15}. In
conjunction with the preceding separator (for low-density objects),
this yields a second proof that the intersection graphs of low-density
objects have polynomial expansion, however with weaker bounds.

A weighted version of the above separator follows by a similar
argument.

\begin{lemma}
  \lemlab{w:s:low:density}%
  Let $\ObjSet$ be a set of $n$ objects in $\Re^d$ with density
  $\cDensity$, and weights $\weightOp: \ObjSet \to \Re$. Let
  $\Weight = \sum_{\obj \in \ObjSet} \weightX{\obj}$ be the total
  weight of all objects in $\ObjSet$.  Then one can compute, in
  expected linear time, a sphere $\sphereC$ that intersects
  $O\pth{\cDensity + \cDensity^{1/d} n^{1-1/d} }$ objects of
  $\ObjSet$. Furthermore, the total weight of objects of $\ObjSet$
  strictly inside/outside $\sphereC$ is at most $c \Weight$, where $c$
  is a constant that depends only on $d$.
\end{lemma}

\begin{proof}
  The argument follows the one used in
  \lemref{separator:low:density}. We pick a representative point from
  each object, and assign it the weigh of the object. Next, we compute
  the smallest ball containing $\geq c \Weight$ of the total weight of
  the points, and the rest of the proof follows readily, observing
  that in the worst case, $n$ objects might be involved in the
  calculations.
\end{proof}

\subsubsection{Divisions}

Consider a set $\Vertices$. A \emphi{cover} of $\Vertices$ is a set
$\clusters = \Set{C_i \subseteq \Vertices}{i=1,\ldots, k \bigr.}$ such
that $\Vertices = \bigcup_{i=1}^k C_i$. A set $C_i \in \clusters$ is a
\emphi{cluster}. A cover of a graph $\defGraph$ is a cover of its
vertices. Given a cover $\clusters$, the \emph{excess} of a vertex
$v \in \Vertices$ that appears in $j$ clusters is $j-1$. The
\emphi{total excess} of the cover $\clusters$ is the sum of excesses
over all vertices in $\Vertices$.

\begin{defn}
  A cover $\Cover$ of $\graph$ is a \emphi{$\exSize$-division} if
  \begin{inparaenum}[(i)]
  \item for any two clusters $C, C' \in \Cover$, the sets
    $C \setminus C'$ and $C' \setminus C$ are separated in $\graph$
    (i.e., there is no edge between these sets of vertices in
    $\graph$), and
  \item for all clusters $C \in \Cover$, we have
    $\cardin{C} \leq \exSize$.
  \end{inparaenum}

  A vertex $v \in \Vertices$ is an \emphi{interior vertex} of a cover
  $\clusters$ if it appears in exactly one cluster of $\clusters$ (and
  its excess is zero), and a \emphi{boundary vertex} otherwise. By
  property (i), the entire neighborhood of an interior vertex of a
  division lies in the same cluster.
\end{defn}

\begin{remark:unnumbered}
  A division is not only a cover of the vertices, but also a cover of
  the edges. Consider a {$\exSize$-division} $\Cover$ of a graph
  $\graph$, and an edge $uv \in \EdgesX{\graph}$. We claim that there
  must be a cluster $C$ in $\Cover$, such that both $u$ and $v$ are in
  $C$. Indeed, if not, then there are two clusters $C_u$ and $C_v$,
  such that $u \in C_u$ and $v \in C_v$, but then
  $ C_u \setminus C_v \ni u$ and $C_v \setminus C_u \ni v$ are not
  separated in $\graph$, contradicting the definition.
\end{remark:unnumbered}

\begin{remark:unnumbered}
  The property of having $\exSize$-divisions is slightly stronger than
  being weakly hyperfinite.  Specifically, a graph is \emph{weakly
    hyperfinite} if there is a small subset of vertices whose removal
  leaves small connected components \cite[Section
  16.2]{no-s-12}. Clearly, $\exSize$-divisions also provide such a set
  (i.e., the boundary vertices). The connected components induced by
  removing the boundary vertices are not only small, but the
  neighborhoods of these components are small as well.
\end{remark:unnumbered}

As noted by Henzinger \etal \cite{hkrs-fspap-97}, strongly sublinear
separators obtain $\exSize$-divisions with total excess $\eps n$ for
$\exSize = \poly(1/\eps)$.  Such divisions were first used by
Frederickson in planar graphs \cite{f-faspp-87}. A proof of the
following known result is provided in \cite{hq-naape-16-arxiv}.
\begin{lemma}
  \lemlab{divisions:small:excess}%
  Let $\graph$ be a graph with $n$ vertices, such that any induced
  subgraph with $m$ vertices has a separator with
  $O(m^{\alpha} \log^\beta m)$ vertices, for some $\alpha < 1$ and
  $\beta \geq 0$.  Then, for $\eps > 0$, the graph $\graph$ has
  $\exSize$-divisions with total excess $\eps n$, where
  \begin{math}
    \exSize = O \pth{ \pth{ \eps^{-1} \log^\beta \eps^{-1}
      }^{1/(1-\alpha)}}.
  \end{math}
  Furthermore, the $\exSize$-division can be computed in polynomial
  time.
\end{lemma}

\begin{remark}[Divisions for weaker separators]
  \remlab{w:div:small:excess}%
  One can still obtain divisions for graph classes with weaker
  separators of size $O\pth{n / \log^{O(1)}n}$. Rather than a
  $\poly(1/\eps)$-division with excess $\eps n$, we get a
  $f(1/\eps)$-division with excess $\epsilon n$ for some function $f$.
  Consequently, the \PTAS throughout this paper extend to a slightly
  broader class of graphs than polynomial expansion, see
  \remref{s:exp:expansion}.
\end{remark}

\begin{corollary}
  \corlab{dev;p:e:l:dense}%
  \begin{inparaenum}[(A)]
  \item \itemlab{p:e:divisions} %
    Let $\graph$ be a graph with polynomial expansion of degree $k$
    and $n$ vertices, and let $\eps > 0$ be fixed. Then $\graph$ has
    $O\pth{(1/\eps)^{2k+2} \log^{2k+1} (1/\eps)}$-divisions with total
    excess $\eps n$.

    \smallskip%
  \item \itemlab{l:d:divisions} %
    Let $\defGraph$ be a $\cDensity$-dense graph with $n$ vertices
    arising out of a given set of objects in $\Re^d$. Then $\graph$
    has $\exSize$-divisions, with $\exSize = O\pth{\cDensity/\eps^d}$
    and total excess at most $\eps n$. This division can be computed
    in $O( n \log( n /\exSize) )$ time.
  \end{inparaenum}
\end{corollary}

\begin{proof}
  (A) By \thmref{p:e:separator}, $\graph$ has separator with
  parameters $\alpha = 1-1/(2k+2)$ and $\beta = 1-1/(2k+2)$.  Plugging
  this into \lemref{divisions:small:excess} implies
  $\exSize$-divisions where
  \begin{math}
    \ds%
    \exSize%
    =%
    O \pth{ \pth{ (1/\eps) \log^\beta (1/\eps) }^{1/(1-\alpha)}}%
    =%
    O \pth{ (1/\eps)^{2k+2} \log^{2k+1} (1/\eps) }.%
  \end{math}

  (B) By \lemref{separator:low:density}, any subgraph of $\graph$ with
  $m$ vertices has a separator of size
  $\leq c \pth{ \cDensity + \cDensity^{1/d} m^{1-1/d}} $, for some
  constant $c$. One can break up $\graph$ in the natural recursive
  fashion using separators (see the proof of
  \lemref{divisions:small:excess} in \cite{hq-naape-16-arxiv} for
  details), until each portion has size $m$, and
  \begin{math}
    c \pth{ \cDensity + \cDensity^{1/d} m^{1-1/d}} \leq \eps m /c',
  \end{math}
  where $c'$ is some absolute constant.  As can be easily verified,
  this holds for $m = \Omega\pth{\cDensity/\eps^d}$. Setting
  $\exSize = m$ implies that the resulting $\exSize$-divisions with
  excess $\leq \eps n$.

  As for the running time, computing the separator for a graph with
  $m$ vertices takes expected $O(m)$ time (assuming basic operation
  like deciding if an object intersects a sphere can be done in
  constant time), using the algorithm of
  \lemref{separator:low:density}, and the recursion depth is
  $O( \log( n /\exSize) )$.
\end{proof}

\subsection{Hereditary and mergeable properties}
\seclab{h:m:prop}

Let $\Prop \subseteq 2^{\Vertices}$ be a property defined over subsets
of vertices of a graph $\graph = (\Vertices,\Edges)$ (e.g., $\Prop$ is
the set of all independent sets of vertices in $\graph$). The property
$\Prop$ is \emphi{hereditary} if for any
$\SetA \subseteq \SetB \subseteq \Vertices$, if $\SetB$ satisfies
$\Prop$, then $\SetA$ satisfies $\Prop$.  The property $\Prop$ is
\emphi{mergeable} if for any $\SetA,\SetB \subseteq \Vertices$ that
are separate in $\graph$, if $\SetA$ and $\SetB$ each satisfy $\Prop$,
then $\SetA \cup \SetB$ satisfies $\Prop$. We assume that whether or
not $\SetA \in \Prop$ can be checked in polynomial time.

Given a set $\Family$ and a property $\Prop \subseteq 2^{\Family}$,
the \emphi{packing problem} associated with $\Prop$, asks to find the
largest subset of $\Family$ satisfying $\Prop$.

\begin{example}
  \exmlab{geometric:packing}%
  Some geometric flavors of packing problems that corresponds to
  hereditary and mergeable properties include:
  \begin{compactenum}[\quad(A)]
  \item Given a collection of objects $\ObjSet$, find a maximum
    independent subset of $\ObjSet$.

  \item Given a collection of objects $\ObjSet$, find a maximum subset
    of $\ObjSet$ with density at most $\cDensity$, where $\cDensity$
    is prespecified.

  \item Find a maximum subset of $\ObjSet$ whose intersection graph is
    planar or otherwise excludes a graph minor.

  \item Given a point set $\PntSet$, a constant $k$, and a collection
    of objects $\ObjSet$, find the maximum subset of $\ObjSet$ such
    that each point in $\PntSet$ is contained in at most $k$ objects
    in $\ObjSet$.
  \end{compactenum}
\end{example}


\section{Approximation algorithms}
\seclab{approx:algorithms}%

\subsection{Approximation algorithms using separators}
\seclab{approx:v:separators}

Graphs whose induced subgraphs have sublinear and efficiently
computable separators are already strong enough to yield \PTAS for
mergeable and hereditary properties (see \secref{h:m:prop} for
relevant definitions). Such algorithms are relatively easy to derive,
and we describe them as a contrast to subsequent results, where such
an approach no longer works. As the following testifies, one can
$(1-\eps)$-approximate, in polynomial time, the independent set in a
low-density or polynomial-expansion graphs (as independent set is a
mergeable and hereditary property).

\begin{lemma}
  \lemlab{ptas:h:prop}%
  Let $\defGraph$ be a graph with $n$ vertices, with the following
  properties: %
  \smallskip%
  \begin{compactenum}[\quad(A)]
  \item Any induced subgraph of $\graph$ on $m$ vertices has a
    separator with $O(m^{\alpha} \log^\beta m)$ vertices, for some
    constants $\alpha < 1$ and $\beta \geq 0$, and this separator can
    be computed in polynomial time. (I.e., low density and polynomial
    expansion graphs have such separators.)

  \item There is a hereditary and mergeable property $\Prop$ defined
    over subsets of vertices of $\graph$.

  \item The largest set $\Opt \in \Prop$, is of size at least $n/c$,
    where $c$ is some absolute constant.
  \end{compactenum}%
  \smallskip%
  Then, for any $\eps > 0$, one can compute, in
  $O\pth{ n^{O(1)} + 2^\exSize \exSize^{O(1)} n }$ time, a set
  $\SetX \in \Prop$ such that
  $\cardin{\SetX} \geq (1-\eps)\cardin{\Opt}$, where
  \begin{math}
    \exSize = O \pth{ \pth{ \eps^{-1} \log^\beta \eps^{-1}
      }^{1/(1-\alpha)}}.
  \end{math}
\end{lemma}

\begin{proof}
  Set $\delta = \eps/2c $. By \lemref{divisions:small:excess}, one can
  compute a $\exSize$-division for $\graph$ in polynomial time, such
  that its total excess is
  \begin{math}
    \excess \leq \delta n \leq \eps n /2c \leq \eps \cardin{\Opt}/2,
  \end{math}
  where $\exSize$ is as stated above.  Throw away all the boundary
  vertices of this division, which discards at most
  $2\excess \leq \eps \cardin{\Opt}$ vertices. The remaining clusters
  are separated from one another, and have size at most $\exSize$. For
  each cluster, we can find its largest subset with property $\Prop$
  by brute force enumeration in $O\pth{ 2^\exSize \exSize^{O(1)}}$
  time per cluster.  Then we merge the sets computed for each cluster
  to get the overall solution. Clearly, the size of the merged set is
  at least $\cardin{\Opt} - 2\excess \geq (1-\eps)\cardin{\Opt}$.  The
  overall running time of the algorithm is
  $O\pth{ n^{O(1)} + 2^\exSize \exSize^{O(1)} n }$.
\end{proof}

\begin{example}[Largest induced planar subgraph]
  Consider a graph $\defGraph$ with $n$ vertices and with polynomial
  expansion of order $k$. Assume, that the task is to find the largest
  subset $\SetX \subseteq \Vertices$, such that the induced subgraph
  $\GInduced{\SetX}$ is, say, a planar graph.  Clearly, this property
  is hereditary and mergeable, and checking if a specific induced
  subgraph is planar can be done in linear time \cite{ht-ept-74}.

  By \obsref{bounded:expansion:degenerate}, the graph $\graph$ is
  $t$-degenerate, for some $t=O(1)$, since $\graph$ has a polynomial
  expansion.  Consequently, $\graph$ contains an independent set of
  size $\geq n/(t+1) =\Omega(n)$. This independent set is a valid
  induced planar subgraph of size $ \Omega(n)$.  Thus, the algorithm
  of \lemref{ptas:h:prop} applies, resulting in an
  $(1-\eps)$-approximation to the largest induced planar subgraph. The
  running time of the resulting algorithm is
  $n^{O(1)} + f(k,\epsilon) n$, for some function $f$.%
\end{example}

\begin{lemma}
  \lemlab{indep:easy}%
  Let $\eps > 0$ be a parameter, and $\ObjSet$ be a given set of $n$
  objects in $\Re^d$ that are $\cDensity$-dense. Then one compute a
  $(1-\eps)$-approximation to the largest independent set in
  $\ObjSet$. The running time of the algorithm is
  $O\pth{ n \log n + 2^\exSize \exSize^{O(1)} n }$, where
  $\exSize = O\pth{\cDensity^{d+1}/\eps^{d}}$.

  More generally, one can compute, with the same running time, an
  $(1-\eps)$-approximate solution for all the problems described in
  \exmref{geometric:packing}.
\end{lemma}
\begin{proof}
  Consider the intersection graph $\graph = \IGraph{\ObjSet}$, and
  observe that by the low-density property, it always have a vertex of
  degree $< \cDensity$ (i.e., take the object in $\ObjSet$ with the
  smallest diameter). As such, removing this object and its neighbors
  from the graph, adding it to the independent set and repeating this
  process, results in an independent set in $\graph$ of size
  $n / \cDensity$. Thus implying that the largest independent set has
  size $\Omega(n)$.  Now, apply the algorithm of \lemref{ptas:h:prop}
  to $\graph$ using the improved $\exSize$-divisions of
  \corref{dev;p:e:l:dense} \itemref{l:d:divisions}.  Here, we need the
  total excess to be bounded by $(\eps /\cDensity) n$, which implies
  that
  \begin{math}
    \exSize%
    =%
    O\bigl( \cDensity/\pth{\eps/\cDensity}^{d} \bigr)%
    =%
    O\pth{\cDensity^{d+1}/\eps^{d}}.
  \end{math}

  For the second part, observe that all the problems mentioned in
  \exmref{geometric:packing} have solution bigger than the independent
  set of $\ObjSet$, and the same algorithm applies with minor
  modifications.
\end{proof}

\begin{remark:unnumbered}
  For computing the largest independent set, one does not need to
  assume the low density on the input -- a more elaborate algorithm
  works, see \lemref{indep:low:density:output} below.

  It is tempting to try and solve problems like dominating set on
  polynomial-expansion graphs using the algorithm of
  \lemref{indep:easy}.  However, note that a dominating set in such a
  graph (or even in a star graph) might be arbitrarily smaller than
  the size of the graph. Thus, having small divisions is not enough
  for such problems, and one needs some additional structure.
\end{remark:unnumbered}

\subsection{Local search for independent set %
  and packing problems}

Chan and Har-Peled \cite{ch-aamis-12} gave a \PTAS for independent set
with planar graphs, and the algorithm and its underlying argument
extends to hereditary graph classes with strongly sublinear separators
(see also the work by Mustafa and Ray \cite{mr-irghs-10}).

\subsubsection{Definitions}

Let $\Prop$ be a hereditary and mergeable property, and let $\exSize$
be a fixed integer.  For two sets, $\SetA$ and $\SetB$, their
\emph{symmetric difference} is
\begin{math}
  \SetA \SetDiff \SetB%
  =%
  \pth{\SetA \setminus \SetB} \cup \pth{\SetB \setminus \SetA}.
\end{math}
Two vertex sets $\SetA$ and $\SetB$ are \emphi{$\exSize$-close} if
$\cardin{\SetA \SetDiff \SetB} \leq \exSize$; that is, if one can
transform $\SetA$ into $\SetB$ by adding and removing at most
$\exSize$ vertices from $\SetA$. A vertex set $\SetA \in \Prop$ is
\emphi{$\exSize$-locally optimal} in $\Prop$ if there is no
$\SetB \in \Prop$ that is $\exSize$-close to $\SetA$ and ``improves''
upon $\SetA$. In a maximization problem $\SetB$ \emphi{improves}
$\SetA$ $\iff$ $\cardin{\SetB} > \cardin{\SetA}$. In a minimization
problem, an improvement decreases the cardinality.

\subsubsection{The local search algorithm in detail}
\seclab{l:s:alg}%

The \emphi{$\exSize$-local search algorithm} starts with an
arbitrary (and potentially empty) solution $\SetA \in \Prop$ and, by
examining all $\exSize$-close sets, repeatedly makes $\exSize$-close
improvements until terminating at a $\exSize$-locally optimal
solution. Each improvement in a maximization (resp., minimization)
problem increases (resp., decreases) the cardinality of the set, so
there are at most $n$ rounds of improvements, where $n$ is the size of
the ground set of $\Prop$. Within a round we can exhaustively try all
exchanges in time $n^{O(\exSize)}$, bounding the total running time by
$n^{O(\exSize)}$.

\subsubsection{Analysis of the algorithm}

\begin{theorem}
  \thmlab{independent}%
  Let $\defGraph$ be a given graph with $n$ vertices, and let $\Prop$
  be a hereditary and mergeable property defined over the vertices of
  $\graph$ that can be tested in polynomial time, Furthermore, let
  $\eps >0$ and $\exSize$ be parameters, and assume that for any two
  sets $\SetX, \SetY \subseteq \Vertices$, such that
  $\SetX, \SetY \in \Prop$, we have that $\GInduced{\SetX \cup \SetY}$
  has a $\exSize$-division with total excess
  $\eps \cardin{\SetX \cup \SetY}$.  Then, the $\exSize$-local search
  algorithm computes, in $n^{O(\exSize)}$ time, a
  $(1-2\eps)$-approximation for the maximum size set
  $\SetZ \subseteq \Vertices$ satisfying $\SetZ \in \Prop$.
\end{theorem}

\begin{proof}
  Let $\Opt \subseteq \Vertices$ be an optimal maximum set satisfying
  $\Prop$, and $\locSol$ be a $\exSize$-locally maximal set satisfying
  $\Prop$. Consider the induced subgraph
  $\graphB = \IGraph{\Opt \cup \locSol}$, and observe that, by
  assumption, there exists a $\exSize$-division
  $\clusters = \setof{\cluster_1,\dots,\cluster_m}$ of $\graphB$, with
  boundary vertices $\BVertices$ and
  \begin{math}
    \excessof{\clusters} \leq \eps \cardin{\Opt \cup \locSol} \leq 2
    \eps \cardin{\Opt}.
  \end{math}
  For $i = 1,\dots,m$, let%

  \centerline{%
    \begin{tabular}{lll}
      $\Opt_i = (\Opt \cap \cluster_i) \setminus \BVertices$,
      &&
         $\oSize_i = \cardin{\Opt_i}$,\\
      $\locSol_i = (\locSol \cap \cluster_i)$,%
      &&
         $\lSize_i = \cardin{\locSol_i}$,\\
      $\BVertices_i = \BVertices \cap \cluster_i$, %
      &
        and
      &
        $\bSize_i = \cardin{\BVertices_i}$.
    \end{tabular}%
  }%
  \medskip%
  Fix $i$, and consider the set
  $\locSol' = (\locSol \setminus \locSol_i) \cup \Opt_i$. Since
  $\cardin{\cluster_i} \leq \exSize$, and $\locSol'$ is
  $\exSize$-close to $\locSol$.  Since $\Prop$ is hereditary,
  $\locSol \setminus \locSol_i \in \Prop$, and since $\Opt_i$ and
  $\locSol \setminus \locSol_i$ are separated, their union $\locSol'$
  is in $\Prop$. Thus, the local search algorithm considers the valid
  exchange from $\locSol$ to $\locSol'$. As the set $\locSol$ is
  $\exSize$-locally optimal, the exchange replacing $\locSol_i$ by
  $\Opt_i$ can not increase the overall cardinality.  Since
  \begin{math}
    \cardin{\locSol} - \lSize_i + \oSize_i %
    = %
    \cardin{\locSol'} %
    \leq %
    \cardin{\locSol}, %
  \end{math}
  this implies that $\lSize_i \geq \oSize_i$. Summing over all $i$, we
  have
  \begin{align*}
    \cardin{\locSol} %
    \geq %
    \sum_{i=1}^m \lSize_i %
    - %
    \sum_{i=1}^m \bSize_i %
    \geq %
    \sum_{i=1}^m \oSize_i %
    - %
    \sum_{i=1}^m \bSize_i %
    \geq %
    \cardin{\Opt} - \excessof{\clusters} %
    \geq %
    (1-2\eps) \cardin{\Opt},
  \end{align*}
  as desired.
\end{proof}

\begin{remark}
  It is illuminating to consider the requirements to make the argument
  of \thmref{independent} go through. We need to be able to break up
  the conflict graph between the local and optimal solutions into
  small clusters, such that the total number of boundary vertices
  (counted with repetition) is small. Surprisingly, even if all (or
  most of) the vertices of a single cluster are boundary vertices, the
  argument still goes through.
\end{remark}%

\begin{lemma}
  \lemlab{indep:low:density:output}%
  Let $\eps > 0$ and $\cDensity$ be parameters, and let $\ObjSet$ be a
  given collection of objects in $\Re^d$ such that any independent set
  in $\ObjSet$ has density $\cDensity$. Then the local search
  algorithm computes a $(1-\eps)$-approximation for the maximum size
  independent subset of $\ObjSet$ in time
  $n^{O(\cDensity/\eps^{d+1})}$.
\end{lemma}

\begin{proof}
  Observing that the union of two $\cDensity$-dense sets results in a
  $2\cDensity$-dense set, and using the algorithm of
  \thmref{independent}, together with the divisions of
  \corref{dev;p:e:l:dense} \itemref{l:d:divisions}, implies the
  result.
\end{proof}

\begin{remark:unnumbered}
  (A) We emphasize that \lemref{indep:low:density:output} requires
  only that independent sets of the input objects $\ObjSet$ have low
  density -- the overall set $\ObjSet$ might have arbitrarily large
  density.

  (B) All the problems of \exmref{geometric:packing} have a \PTAS
  using the \lemref{indep:low:density:output} as long as the output
  has low density.
\end{remark:unnumbered}

\subsection{Dominating Set} %
\seclab{dom:set}%

We are interested in approximation algorithm for the following
generalization of the dominating set problem.

\begin{defn}
  Let $\defGraph$ be an undirected graph, and let $\DomSet$ and
  $\CovSet$ be two subsets of $\Vertices$.  The set $\DomSet$
  \emph{dominates} $\CovSet$ if every vertex in $\CovSet$ either is in
  $\DomSet$ or is adjacent to some vertex in $\DomSet$. In the
  \emphi{dominating subset problem}, one is given an undirected graph
  $\defGraph$ and two subsets of vertices $\CovSet$ and $\DomSet$,
  such that $\DomSet$ dominates $\CovSet$. The task is to compute the
  smallest subset of $\DomSet$ that still dominates $\CovSet$.
\end{defn}

One can approximate the dominated set, as the reader might expect, via
a local search algorithm.  Before analyzing the algorithm, we need to
develop some tools to be able to argue about the interaction between
the local and optimal solution.

\subsubsection{Shallow packings}

\begin{defn}
  \deflab{shallow:packing}%
  Given a graph $\defGraph$, a collection of sets
  $\Family = \Set{\cluster_i \subseteq \Vertices}{i=1,\ldots,t}$ is a
  \emphi{$(\iCov,t)$-shallow packing} of $\graph$, or just a
  \emphi{$(\iCov,t)$-packing}, if for all $i$, the induced graph
  $\GInduced{\cluster_i}$ is $t$-shallow (see \defref{t:shallow}), and
  every vertex of $\Vertices$ appears in at most $\iCov$ sets of
  $\Family$\footnote{We allow a set $C$ to appear in $\Family$ more
    than once; that is, $\Family$ is a multiset.}.

  The \emphi{induced packing graph} $\ICovGraph{\graph}{\Family}$ has
  $\Family$ as the set of vertices, and two clusters
  $\cluster, \cluster' \in \Family$ are connected by an edge if they
  share a vertex (i.e., $\cluster \cap \cluster' \neq \emptyset$), or
  there are vertices $u \in \cluster$ and $v \in \cluster$, such that
  $uv \in \Edges$.
\end{defn}

For example, the induced packing graph of a $(1,t)$-packing is the
$t$-shallow minor induced by the clusters of the packing (see
\defref{shallow:minor}).%

\begin{lemma}
  \lemlab{edge-density-shallow-cover}%
  Let $\defGraph$ be an undirected graph, and $\Family$ an
  $(\iCov,t)$-packing of $\graph$. Then the induced packing graph
  $\graphA = \ICovGraph{\graph}{\Family}$ has edge density
  \begin{math}
    \mytfrac{\cardin{\EdgesX{\graphA}}}
    {\cardin{\verticesof{\graphA}}}%
    \leq %
    2 (t+1)^2 \iCov^2 \gradY{t}{\graph} + \iCov,
  \end{math}
  where $\gradY{t}{\graph}$ is the $t$-shallow density of $\graph$,
  see \defref{r:shallow:density}.
\end{lemma}
\begin{proof}
  Let the clusters of $\Family$ be
  $\setof{\clusterZ{1},\dots,\clusterZ{m}}$. For each cluster
  $\clusterZ{i} \in \Family$, designate a center vertex
  $\cvX{i} \in \clusterZ{i}$ that can reach any other vertex in
  $\clusterZ{i}$ by a path contained in $\clusterZ{i}$ of length $t$
  or less.  Let $\pi:\IntRange{m} \to \IntRange{m}$ be a random
  permutation of the cluster indices, chosen uniformly at random, and
  initialize $\scoops = \emptyset$, where
  $\IntRange{m} = \brc{1,\ldots, m}$.  For $i = 1,\dots,m$, in order,
  check if $\rcvX{i}$ has been ``scooped''; that is, if
  \begin{math}
    \rcvX{i} \in \bigcup_{\cluster' \in \scoops} \cluster',
  \end{math}
  and if so, ignore it. Otherwise, let $\rsX{i}$ be the set of
  vertices of the connected component of $\rcvX{i}$ in the induced
  subgraph of $\graph$ over
  \begin{math}
    \randomcluster{i} \setminus {\bigcup_{\cluster' \in \scoops}
      \cluster'},
  \end{math}
  and add $\rsX{i}$ to $\scoops$.  Intuitively, the set $\scoops$ is a
  $(1,t)$-packing of $\graph$ resulting from randomly shrinking the
  clusters of $\Family$.

  \smallskip%
  We bound the number of edges in
  $\graphA = \ICovGraph{\graph}{\Family}$ by a function of the
  expected number of edges in the random graph
  $\graphA' = \ICovGraph{\graph}{\scoops}$.  Let
  \begin{math}
    \smalledges %
    = %
    \Set{ %
      \clusteredge{i}{j} %
      \in %
      \EdgesX{\graphA} %
    }{ %
      \cvX {i} \in \clusterZ{j} %
      \text{ or } \cvX {j} \in \clusterZ{i} 
    } %
  \end{math}
  be the set of edges between pairs of clusters where the center of
  one cluster is also in the other cluster. Since a center $\cvX{i}$
  can be covered at most $\iCov$ times by $\Family$, we have
  $\cardin{\smalledges} \leq \iCov \cardin{\Family}$.  Next, consider
  the set of remaining edges,
  \begin{align*}
    \bigedges%
    = %
    \Set{ \clusteredge{i}{j} \in \EdgesX{\graphA} } %
    { \cvX{j} \notin \clusterZ{i} %
    \text{ and } %
    \cvX{i} \notin \clusterZ{j} %
    },%
  \end{align*}
  between adjacent clusters where neither center lies in the opposing
  cluster. For an edge $\clusteredge{i}{j} \in \bigedges$, consider
  the probability that $\scoopedge{i}{j} \in \EdgesX{\graphA'}$.

  Since $\clusterZ{i}$ and $\clusterZ{j}$ are adjacent in $\graphA$,
  there is a path $P$ in $\graph$ from $\cvX{i}$ to $\cvX{j}$ of
  length at most $2t + 1$ that is contained in
  $\clusterZ{i} \cup \clusterZ{j}$, and a sufficient condition for
  $\scoopedge{i}{j} \in \EdgesX{\graphA'}$ is that $P$ is contained in
  $\scoop{i} \cup \scoop{j}$. This holds if the permutation $\pi$
  ranks $i$ and $j$ ahead of any other index $k$ such that
  $\clusterZ{k}$ intersects the vertices of $P$. There are at most
  $2t + 2$ vertices on $P$, where each vertex can appear in at most
  $\iCov$ clusters of $\Family$, and overall there are at most
  \begin{align*}
    \ell = 2 (t+1) \iCov
  \end{align*}
  clusters that compete for control over the vertices of $P$ in
  $\scoops$. The probability that, among these relevant clusters, the
  random permutation $\pi$ ranks $i$ and $j$ before all others is
  \begin{math}
    2(\ell-2)!/\ell! \geq 2 / \ell^2.
  \end{math}
  Therefore, for $\clusteredge{i}{j} \in \bigedges$, we have
  \begin{math}
    \Prob{ %
      \scoopedge{i}{j} \in \EdgesX{\graphA'} %
    }%
    \geq %
    2/\ell^2. %
  \end{math}
  By linearity of expectation, and since
  $\graphA' = \ICovGraph{\graph}{\scoops}$ is a $t$-shallow minor of
  $\graph$, we have
  \begin{align*}
    \cardin{\bigedges} %
    &= %
      \sum_{e \in \bigedges} %
      \frac{\ell^2 / 2}{\ell^2 / 2} %
      \leq %
      \frac{\ell^2}{2} \sum_{e \in \bigedges} %
      \Prob{\bigl. e \in \EdgesX{\graphA'}} %
      = %
      \frac{\ell^2}{2} \Ex{\bigl. \cardin{ \EdgesX{\graphA'}}} %
          \leq %
          \frac{\ell^2}{2} \gradY{t}{\graph} \cardin{\scoops} %
          \leq %
          \frac{\ell^2}{2} \gradY{t}{\graph} \cardin{\Family}.
  \end{align*}
  We conclude that
  \begin{math}
    \Bigl.  \mytfrac{\cardin{\EdgesX{\graphA}}}
    {\cardin{\verticesof{\graphA}}}%
    = %
    \mytfrac{\cardin{\bigedges}}{\cardin{\Family}} + %
    \mytfrac{\cardin{\smalledges}}{\cardin{\Family}} %
    \leq %
    (\ell^2 / 2) \gradY{t}{\graph} + \iCov,
  \end{math}
  as desired.
\end{proof}

\begin{remark:unnumbered}
  Results similar to \lemref{edge-density-shallow-cover} are already
  known \cite{no-s-12}. However, our result has a polynomial
  dependency on $\iCov$ and $t$, while the known results seems to
  imply an exponential dependency.
\end{remark:unnumbered}%

\begin{lemma}
  \lemlab{expansion:shallow:cover} %
  Consider a graph $\graph$, and an $(\iCov,t)$-packing $\Family$ of
  $\graph$. Then, for any integer $u >0$, we have
  \begin{align*}
    \gradY{u}{\bigl.\icgA} %
    \leq %
    5 \iCov^2 (2u + 1)^2 (2t + 1)^2 \gradY{2 tu + t + u}{\graph}.
  \end{align*}
  In particular, if $t$ and $\iCov$ are constants, and $\graph$ has
  polynomial of order $k$, then $\icgA$ has polynomial expansion of
  order $k + 2$.
\end{lemma}

\begin{proof}
  For $u \geq 1$, consider a $u$-shallow minor $\graphA$ of
  $\icgA$. Every cluster in this cover corresponds to an expanded
  cluster in the original graph $\graph$ with radius $2t u + u + t$,
  and a vertex might participate in $\iCov$ such clusters. That is,
  the resulting set $\FamilyA$ of clusters is an
  $(\iCov, 2tu + u + t)$-packing of $\graph$.  By
  \lemref{edge-density-shallow-cover}, we have
  \begin{align*}
    \alpha(\graphA)%
    &=%
      \frac{\cardin{\EdgesX{\graphA}}}{\cardin{\verticesof{\graphA}}}%
      =%
      \edgedensityof{\bigl.\ICovGraph{\graph}{\FamilyA}}%
      \leq %
      5 \iCov^2 \pth{\bigl. 4tu + 2u + 2t + 1}^2 \gradY{2tu + t +
      u}{\graph} %
    \\%
    &\leq %
      5 \iCov^2 (2t + 1)^2(2u + 1)^2 \gradY{2tu + t + u}{\graph}.
  \end{align*}
  By \defref{r:shallow:density}, we have
  \begin{math}
    \gradY{u}{\bigl.\icgA} %
    = %
    \sup_{\graphA \in \minorsDY{u}{\icgA}} \alpha(\graphA) %
    \leq %
    5 \iCov^2 (2t + 1)^2(2u + 1)^2 \gradY{2tu + t + u}{\graph} .
  \end{math}
\end{proof}

\subsubsection{Lexical product and shallow density}

An interesting consequence of the above is an improvement over known
bounds for the shallow density under lexical product (this result is
not required for the rest of the paper).  Given two graphs $\graph$
and $\graphA$, the \emphi{lexical product} $\graph \lexprod \graphA$
is the graph obtained by blowing up each vertex in $\graph$ with a
copy of $\graphA$. More formally, $\graph \lexprod \graphA$ has vertex
set $\VerticesX{\graph} \times \VerticesX{\graphA}$ and an edge
between two vertices $(x,y)$ and $(x',y')$ if either (a)
$\edgeY{x}{x'} \in \EdgesX{\graph}$, or (b) $x = x'$ and
$\edgeY{y}{y'} \in \EdgesX{\graphA}$.
\begin{corollary}
  \corlab{grad-lexical-product}%
  For any graph $\graph$, clique $K_{\iCov}$, and
  $t \in \naturalnumbers$, we have
  \begin{math}
    \gradY{t}{\graph \lexprod K_{\iCov}} %
    \leq %
    5 \iCov^2 (t+1)^2 \gradY{t}{\graph}.
  \end{math}
  In particular, if $\iCov$ is constant and $\graph$ has polynomial
  expansion of order $k$, then $\graph \lexprod K_{\iCov}$ has
  polynomial expansion of order $k+2$.
\end{corollary}
\begin{proof}
  A $t$-shallow minor of $\graph \lexprod K_{\iCov}$ is the induced
  packing graph of the $(\iCov,t)$-packing formed by its
  clusters. Thus, the claimed inequality follows from
  \lemref{edge-density-shallow-cover}.
\end{proof}

\corref{grad-lexical-product} is an exponential improvement over the
best previously known bounds, on the order of
\begin{math}
  \gradY{t}{\graph \lexprod K_{\iCov}}%
  \leq%
  \pth{\Bigl. O\pth{\big. \iCov t \gradY{t}{\graph} } }^{O(t)},
\end{math}
by \Nesetril and Ossona \si{de} Mendez \cite{no-gcbe1-08} (see also
the comments following the proof of Proposition 4.6 in
\cite{no-s-12}).


\subsubsection{Low density objects and $(\iCov,t)$-packing{}s}
\seclab{l:d:o:packing}

\begin{defn}
  For a set of objects $\ObjSet$, a collection of subsets
  \begin{math}
    \Family = \Set{\cluster_i \subseteq \ObjSet}{i=1,\ldots,t}
  \end{math}
  forms a \emphi{$(\iCov,t)$-shallow packing} of $\graph$ if, for all
  $i$, the intersection graph $\IGraph{\cluster_i}$ is $t$-shallow
  (see \defref{t:shallow}), and every object of $\ObjSet$ appears in
  at most $\iCov$ sets of $\Family$.  The \emphi{induced object set}
  $\IObjSet{\ObjSet}{\Family}$ is the collection of objects
  $\Set{\bigcup_{\obj \in \cluster_i} \obj}{\cluster_i \in \Family}$
  formed by taking the union of each cluster in $\Family$.
\end{defn}

\begin{lemma}%
  \lemlab{shallow:cover:objects}%
  Let $\ObjSet$ be a collection of objects with density $\cDensity$ in
  $\Re^d$, and let $\Family$ be an $(\iCov,t)$-shallow packing. Then
  the induced object set $\IObjSet{\ObjSet}{\Family}$ has density
  $O(\iCov \cDensity t^{O(d)})$.
\end{lemma}

\begin{proof}
  Consider the multiset of objects
  $\ObjSetA = \bigcup_{\cluster_i \in \Family} \cluster_i$ where each
  object $f \in \ObjSet$ is repeated according to its multiplicity in
  $\Family$. Since each object in $\ObjSet$ appears in $\ObjSetA$ at
  most $\iCov$ times, $\ObjSetA$ has density $\iCov \cDensity$.  Every
  cluster $\cluster \in \Family$ can be interpreted as a new cluster
  $f(\cluster)$ of objects of $\ObjSetA$, where the resulting set of
  clusters $\Family' = \Set{ f(\cluster) }{\cluster \in \Family}$ are
  now disjoint.

  As such, $\IObjSet{\ObjSetA}{\Family'}$ is a $t$-shallow minor of
  $\ObjSetA$. By \lemref{density:shallow:minors}, the graph
  $\IObjSet{\ObjSetA}{\Family'}$ has density
  $O(\iCov \cDensity t^{O(d)})$.
\end{proof}%

\subsubsection{The result}

Shallow packings arise in the analysis of the approximation algorithm
for dominating set, where vertices are clustered together around the
the vertex that dominates them. In this setting, we prefer the
following simple and convenient terminology.

\begin{defn}
  \deflab{flower:head}%
  Given a dominating set $\DomSet = \brc{v_1, \ldots, v_m}$ of
  vertices in a graph $\defGraph$, and a set of vertices
  $\CovSet \subseteq \Vertices$ being dominated by $\DomSet$, we
  generate a sequence of clusters
  $\cluster_1,\dots,\cluster_m \subseteq \DomSet \cup \CovSet$ that
  specifies for every element of $\DomSet$, which elements it covers.

  Initially, we set $\DomSet_0 = \DomSet$ and $\CovSet_0 = \CovSet$.
  In the $i$\th iteration, for $i =1,\ldots,m$, let
  \begin{align*}
    \cluster_i%
    =%
    \brc{v_i} \cup \pth{\big.\pth{\NbrX{v_i} \cap \CovSet_{i-1}}
    \setminus \DomSet_{i-1}},
    \qquad%
    \DomSet_i = \DomSet_{i-1} \setminus \brc{v_i}, \qquad \text{ and
    }%
    \CovSet_i = \CovSet_i \setminus \cluster_i,
  \end{align*}
  where $\NbrX{v_i}$ is the set of vertices adjacent to $v_i$ in
  $\graph$.  Conceptually, $\cluster_i$ induces a star-like graph
  $\graph_i$ over $\cluster_i$, where every vertex of $\cluster_i$ is
  connected to $v_i$. The cluster $\cluster_i$ (and implicitly to
  $\graph_i$) is a \emphi{flower}, where $v_i$ is its \emphi{head}.
  The collection of clusters
  $\FDecomp{\DomSet}{\CovSet} = \brc{ \cluster_1, \ldots, \cluster_m}$
  is the \emphi{flower decomposition} of the given instance.  Note
  that a flower is a $1$-shallow graph, and a flower decomposition is
  a $(1,1)$-shallow packing.
\end{defn}

\begin{theorem} %
  \thmlab{ptas:subset:dom} %
  Let $\defGraph$ be a graph with $n$ vertices and with polynomial
  expansion of order $k$, let $\CovSet, \DomSet\subseteq \Vertices$ be
  two sets of vertices such that $\DomSet$ dominates $\CovSet$, and
  let $\eps > 0$ be fixed. Then, for
  \begin{math}
    \exSize = O\pth{\eps^{-(2k+6)}\log^{2k+5} (1/\eps)},
  \end{math}
  the $\exSize$-local search algorithm computes, in
  $n^{O\pth{\exSize}}$ time, a $(1 + \eps)$-approximation for the
  smallest cardinality subset of $\DomSet$ that dominates $\CovSet$.
\end{theorem}
\begin{proof}
  The algorithm starts with the whole collection $\DomSet$ as the
  local solution, and performs legal local exchanges
  of size $\exSize$ that decrease the size of the local solution by at
  least one until no such exchange is available (see
  \secref{l:s:alg}).

  Let $\Opt \subseteq \DomSet$ and $\locSol \subseteq \DomSet$ be the
  optimal and locally minimal sets dominating $\CovSet$, respectively.
  Let $\optFSet = \FDecomp{\Opt}{\CovSet}$ and
  $\locFSet = \FDecomp{\locSol}{\CovSet}$ be the corresponding flower
  decompositions. In the following, for vertices $\opnt \in \Opt$ and
  $\lpnt \in \locSol$, we use $\optFl{\opnt}$ and $\locFl{\lpnt}$ to
  denote their flower in these decompositions, respectively.

  Let $\graphA = \ICovGraph{\graph}{{\optFSet \cup \locFSet} }$ be the
  induced packing graph of $\Family = \optFSet \cup \locFSet$.  The
  set $\Family$ is a $(2,1)$-shallow cover of $\graph$, and
  \lemref{expansion:shallow:cover} implies that $\graphA$ has
  polynomial expansion of order $k+2$.  By \corref{dev;p:e:l:dense}
  \itemref{p:e:divisions}, $\graphA$ has
  $\exSize = O\pth{(1/\eps)^{2k+6} \log^{2k+5} (1/\eps)}$-division
  \begin{math}
    \clusters = \setof{\cluster_1,\dots,\cluster_m}
  \end{math}
  with a set of boundary vertices $\BVertices$, and total excess
  \begin{math}
    (\eps/4) \cardin{\Family}%
    \leq%
    (\eps/4)\pth{\bigl. \cardin{\optFSet} + \cardin{ \locFSet}} %
    \leq%
    (\eps/2) \cardin{\locSol}.
  \end{math}
  For $i = 1,\dots,m$, let %
  \smallskip
  \begin{compactenum}[\qquad(i)]
  \item
    \begin{math}
      \Opt_i = \Set{\opnt \in \Opt}{ \optFl{\opnt} \in \optFSet \cap
        \cluster_i \bigr.},
    \end{math}

  \item
    $\locSol_i = \Set{\lpnt \in \locSol}{\locFl{\lpnt} \in
      \pth{\locFSet \cap \cluster_i} \setminus \BVertices
      \bigr.}\Bigr.$, and %

  \item $\BVertices_i = \BVertices \cap \cluster_i$.
  \end{compactenum}%

  \smallskip%
  Fix $i$, and consider the set
  $\locSol' = (\locSol \setminus \locSol_i) \cup \Opt_i$.  If a vertex
  $v \in \CovSet$ is not dominated by $\locSol \setminus \locSol_i$,
  then $v \in \locFl{\lpnt} \subseteq \NbrX{\lpnt} \cup \setof{\lpnt}$
  for some $\lpnt \in \locSol_i$, and
  $v \in \optFl{\opnt} \subseteq \NbrX{\opnt} \cup \setof{\opnt}$ for
  some $\opnt \in \Opt$ with $\locFl{\lpnt}$ adjacent to
  $\optFl{\opnt}$ in $\graphA$. The cluster $\locFl{\lpnt}$ is an
  interior vertex of $\cluster_i$, so $\optFl{\opnt}$ must be in the
  cluster $\cluster_i$, and $\opnt \in \Opt_i$. Therefore, the
  alternative solution $\locSol'$ dominates $v$, and overall,
  $\locSol'$ dominates $\CovSet$.

  Since $\locSol$ is $\exSize$-locally minimal, and the exchange size
  is
  \begin{math}
    \cardin{\locSol \SetDiff \locSol'}%
    =%
    \cardin{\locSol_i \cup \Opt_i}%
    \leq%
    \cardin{\cluster_i}%
    \leq%
    \exSize,
  \end{math}
  the new solution $\locSol'$ is at least as large as
  $\locSol$. Expanding
  \begin{math}
    \cardin{\locSol} %
    \leq \cardin{\locSol'}%
    =%
    \cardin{ (\locSol \setminus \locSol_i) \cup \Opt_i}%
    =%
    \cardin{\locSol} - \cardin{\locSol_i} + \cardin{\Opt_i},
  \end{math}
  we have $\cardin{\locSol_i} \leq \cardin{\Opt_i}$. Summed over all
  the clusters $W_i$, we have,
  \begin{align*}
    \cardin{\locSol} %
    \leq %
    \sum_{i=1}^m \pth{\bigl.\cardin{\locSol_i} +
    \cardin{\BVertices_i}} %
    \leq %
    \sum_{i=1}^m \pth{\bigl. \cardin{\Opt_i} +
    \cardin{\BVertices_i}} %
    \leq %
    \cardin{\Opt} + 2 \excessof{\clusters} %
    \leq %
    \cardin{\Opt} + \frac{ \eps}{2} \cardin{\locSol}.
  \end{align*}
  Solving for $\cardin{\locSol}$, we conclude that
  \begin{math}
    \cardin{\locSol} %
    \leq %
    \cardin{\Opt} / (1 - \eps/2) %
    \leq %
    (1+ \eps) \cardin{\Opt},
  \end{math}
  as desired.
\end{proof}

\subsubsection{Extensions -- multi-cover and reach}

One can naturally extend dominating set in the following ways:
\smallskip%
\begin{compactenum}[(A)]
\item \textbf{Demands}: For every $v \in \CovSet$, there is an integer
  $\demandX{v} \geq 0$, which is the \emphi{demand} of $v$; that is,
  $v$ has to be adjacent to at least $\demandX{v}$ vertices in the
  dominating set. In the context of set cover, this is known as the
  multi-cover problem, see \cite{cch-smcpg-12}. Let
  $\MaxDemand = \max_{v \in \CovSet} \demandX{v}$ be the
  \emphi{demand} of the given instance.

\item \textbf{Reach}: Instead of the dominating set being adjacent to
  the vertices that are being covered, for every vertex
  $v \in \CovSet$ one can associate a distance $\reachX{v} \geq 1$ --
  which is the maximum number of hops the dominating vertex can be
  away from $v$ in the given graph. The \emphi{reach} of the given
  instance is
  \begin{math}
    \MaxReach = \max_{v \in \CovSet} \reachX{v}.
  \end{math}
\end{compactenum}
\smallskip%
Thus, a vertex $v$ with demand $\demandX{v}$ and reach $\reachX{v}$,
requires that any dominating set would have $\demandX{v}$ vertices in
edge distance at most $\reachX{v}$ from it.

\begin{lemma} %
  \lemlab{ptas:subset:dom:2} %
  Let $\defGraph$ be a graph with $n$ vertices and with polynomial
  expansion of order $k$, sets $\CovSet\subseteq \Vertices$ and
  $\DomSet \subseteq \Vertices$, such that $\DomSet$ dominates
  $\CovSet$, and let $\eps > 0$ be fixed. Furthermore, assume that for
  each vertex $v \in \CovSet$, there are associated demand and reach,
  where the reach $\MaxReach$ and demand $\MaxDemand$ of the given
  instances are bounded by a constant.

  Then, for $\exSize = O\pth{\eps^{-(2k+6)}\log^{2k+5} (1/\eps)}$, the
  $\exSize$-local search algorithm computes, in $n^{O\pth{\exSize}}$
  time, a $(1 + \eps)$-approximation for the smallest cardinality
  subset of $\DomSet$ that dominates $\CovSet$ under the reach and
  demand constraints.
\end{lemma}

\begin{proof}
  Let $\prec$ be an arbitrary ordering on the vertices of $\graph$.
  For a set of vertices $\SetX \subseteq \Vertices$ and a vertex
  $z \in \Vertices$, let $\nnK{k}{z}{\SetX}$ be the $k$ closest
  vertices to $z$ in $\SetX$, with respect to the length of the
  shortest path in $\graph$, and resolving ties by $\prec$. The
  ordering $\prec$ ensures that $\nnK{k}{z}{\SetX}$ is uniquely
  defined for any vertex in the graph.

  In the following argument, fix a set $\SetX \subseteq \DomSet$ that
  dominates $\CovSet$ and complies with the given constraints, and
  assign every vertex of $u \in \CovSet$ to each of the vertices of
  $\nnK{\demandX{u}}{u}{\SetX}$.  For a vertex $v \in \SetX$, let
  $\ServeX{v}$ be the set of vertices assigned to it.  For each vertex
  $v \in \SetX$, let $\TreeX{v}$ be the minimal subtree of the {BFS}
  tree rooted at $v$ that includes all the vertices of
  $\ServeX{v} \cup \brc{v}$.  The \emph{flower}
  $\cluster_v = \VerticesX{\TreeX{v}}$ is $\MaxReach$-shallow in
  $\graph$.  Let
  \begin{math}
    \Family = \FDecomp{\SetX}{\CovSet} = \Set{ \cluster_v }{v \in
      \SetX}
  \end{math}
  be the resulting \emph{flower decomposition} of $X$.

  We claim that a vertex $z$ of $\graph$ is covered at most
  $\MaxDemand$ times by the flowers of $\Family$. More precisely, we
  prove that $z$ is covered by a flower $\cluster_v$ only if
  $v \in \nnK{\MaxDemand}{z}{\SetA}$. For the sake of contradiction,
  suppose $z \in \cluster_v$ and that
  $v \notin \nnK{\MaxDemand}{z}{\SetA}$. Then $z$ is not assigned to
  $v$, so there must be a vertex $u$ assigned to $z$ and an associated
  shortest path
  \begin{math}
    \trA_{uv} = \trA_{uz} | \trA_{zv}
  \end{math}
  from $u$ to $v$ through $z$, where $\trA_{uz}$ is the subpath from
  $u$ to $z$ and $\trA_{zu}$ is the subpath from $z$ to $v$. Since
  $v \in \nnK{\MaxDemand}{u}{\SetA} \setminus
  \nnK{\MaxDemand}{z}{\SetA}$, and both sets
  $\nnK{\MaxDemand}{u}{\SetA}$ and $\nnK{\MaxDemand}{z}{\SetA}$ have
  the same cardinality, there exists another vertex
  $v' \in \nnK{\MaxDemand}{z}{\SetA} \setminus
  \nnK{\MaxDemand}{u}{\SetA}$.  Let $\trB_{zv'}$ be the shortest-path
  from $z$ to $v'$. By construction of $\nnK{\MaxDemand}{z}{\SetA}$,
  either
  \begin{math}
    \lenX{\trB_{zv'}} < \lenX{\trA_{zv}}, %
  \end{math}
  or
  \begin{math}
    \lenX{\trB_{zv'}} = \lenX{\trA_{zv}} %
  \end{math}
  and $v' \prec v$. This implies that either
  \begin{math}
    \lenX{\trA_{uz} | \trB_{zv'}} %
    < \lenX{\trA_{uz} | \trA_{zv}}, %
  \end{math}
  or $v' \prec v$ and
  \begin{math}
    \lenX{\trA_{uz} | \trB_{zv'}} %
    = %
    \lenX{\trA_{uz} | \trA_{zv}}, %
  \end{math}
  where $|$ denotes concatenation of paths.  In any case, if ties are
  broken by $\prec$, then $v'$ is closer to $u$ than $v$ is, a
  contradiction to the premise that $v \in \nnK{\MaxDemand}{u}{\SetA}$
  and $v' \notin \nnK{\MaxDemand}{u}{\SetA}$. Thus, if $z$ is in a
  flower $\cluster_v$, then $v \in \nnK{\MaxDemand}{z}{\SetA}$.

  Now, consider the local solution $\locSol$ and the optimal solution
  $\Opt$.  Let $\optFSet = \FDecomp{\Opt}{\CovSet}$ and
  $\locFSet = \FDecomp{\locSol}{\CovSet}$ be the flower decompositions
  of the local and optimal solutions, respectively. Each flower
  decomposition includes an element at most $\MaxDemand$ times, so the
  combined collection $\Family = \optFSet \cup \locFSet$ is a
  $\bigl(2\MaxDemand, \MaxReach \bigr)$-shallow packing. By
  \lemref{expansion:shallow:cover}, the induced packing graph
  $\graphA = \ICovGraph{\graph}{\Family}$ has polynomial expansion of
  order $k+2$.  We now follow the argument used in the proof of
  \thmref{ptas:subset:dom}, providing the details for the sake of
  completeness.

  Let
  \begin{math}
    \exSize = O\pth{\eps^{-2k+6}\log^{2k+5} (1/\eps)}.
  \end{math}
  There is a $\exSize$-division of $\graphA$ into clusters
  $\clusterA_1, \ldots, \clusterA_m \subseteq \Family$, with
  $\bdDiv \subseteq \Family$ boundary vertices and total excess
  $\cardin{\bdDiv}\leq (\eps/4) \cardin{\Family}$.  For
  $i = 1,\dots,m$, let %
  \smallskip
  \begin{compactenum}[\qquad(i)]
  \item $\optFSet_i = \optFSet \cap \clusterA_i$,

  \item
    $\locFSet_i = \pth{\locFSet \cap \clusterA_i} \setminus \bdDiv$,
    and %

  \item $\bdDiv_i = \bdDiv \cap \clusterA_i$.
  \end{compactenum}%
  Fix $i$, and consider the cover
  $\locFSet' = (\locFSet \setminus \locFSet_i) \cup \optFSet_i$.
  Consider a vertex $v \in \Vertices$ such that there is a flower in
  $\locFSet \setminus \locFSet'$ that covers it (i.e., the vertex
  ``lost'' coverage in this potential exchange).  This implies that
  $v$ must be covered by a flower $\flower \in \locFSet_i$; that is,
  by a flower that corresponds to a vertex of $\graphA$ that is
  internal to $\clusterA_i$.  Any flower $\flower' \in \Family$ that
  covers $v$ is adjacent to $\flower$ in $\graphA$, by the definition
  of $\graphA$ and as $\flower$ and $\flower'$ share a vertex.  As
  $\flower$ is internal to $\clusterA_i$, all the flowers of $\Family$
  that cover $v$ are in $\clusterA_i$, and in particular, all the
  flowers covering $v$ in the optimal solution belong to
  $\optFSet_i$. Thus, the coverage provided by $\locFSet'$ meets the
  demand and reach requirements of $v$. The rest of the argument now
  follows the proof of \thmref{ptas:subset:dom}.
\end{proof}

\subsubsection{Extension: Connected dominating set}
\seclab{connected:dominating:set}%

The algorithms of \thmref{ptas:subset:dom} and
\lemref{ptas:subset:dom:2} can be extended to handle the additional
constraint that the computed dominating set is also connected.  In
this setting, the local search algorithm only considers beneficial
exchanges that result in a connected dominating set.


\begin{lemma}
  \lemlab{ptas:subset:dom:3}%
  Let $\defGraph$ be a graph with $n$ vertices and polynomial
  expansion of order $k$, and let $\DomSet \subseteq \Vertices$ be a
  connected dominating set. For each vertex $v \in \Vertices$, let
  $\demandX{v} \geq 1$ be its associated demands, and let
  $\MaxDemand = \max_{v \in \Vertices} \demandX{v}$ be bounded by a
  constant (here, the dominating set has to dominate all the vertices
  in the graph). Then, for
  $\exSize = O\pth{\eps^{-(2k+6)}\log^{2k+5} (1/\eps)}$, the
  $\exSize$-local search algorithm computes, in $n^{O\pth{\exSize}}$
  time, a $(1+\eps)$-approximation for the smallest cardinality subset
  of $\DomSet$ that is \emph{connected} and dominates $\Vertices$
  under the demand constraints.
\end{lemma}

\begin{proof}
  We extend the notations and argument used in
  \thmref{ptas:subset:dom}. To recap, let $\Opt \subseteq \DomSet$ and
  $\locSol \subseteq \DomSet$ be the optimal and locally minimum sets
  dominating $\CovSet$, respectively.  Let
  $\optFSet = \FDecomp{\Opt}{\CovSet}$ and
  $\locFSet = \FDecomp{\locSol}{\CovSet}$ be the corresponding flower
  decompositions (see \defref{flower:head}). In the following, for
  vertices $\opnt \in \Opt$ and $\lpnt \in \locSol$, we use
  $\optFl{\opnt}$ and $\locFl{\lpnt}$ to denote their flower in these
  decompositions, respectively.

  Let $\graphA = \ICovGraph{\graph}{{\optFSet \cup \locFSet} }$ be the
  induced packing graph of $\Family = \optFSet \cup \locFSet$. As
  before, we can apply \corref{dev;p:e:l:dense}
  \itemref{p:e:divisions} to $\graphA$ to generate a
  $\exSize = O\pth{(1/\eps)^{2k+6} \log^{2k+5} (1/\eps)}$-division
  \begin{math}
    \clusters = \setof{\cluster_1,\dots,\cluster_m}
  \end{math}
  with a set of boundary vertices $\BVertices$, and total excess
  \begin{math}
    (\eps/4) \cardin{\Family}%
    \leq%
    (\eps/4)\pth{\bigl. \cardin{\optFSet} + \cardin{ \locFSet}} %
    \leq%
    (\eps/2) \cardin{\locSol}.
  \end{math}
  For $i = 1,\dots,m$, let %
  \smallskip
  \begin{compactenum}[\qquad(i)]
  \item
    \begin{math}
      \Opt_i = \Set{\opnt \in \Opt}{ \optFl{\opnt} \in \optFSet \cap
        \cluster_i \bigr.},
    \end{math}

  \item
    $\locSol_i = \Set{\lpnt \in \locSol}{\locFl{\lpnt} \in
      \pth{\locFSet \cap \cluster_i} \setminus \BVertices
      \bigr.}\Bigr.$, and %

  \item $\BVertices_i = \BVertices \cap \cluster_i$.
  \end{compactenum}%
  Fix $i$, and consider the set
  $\locSol' = (\locSol \setminus \locSol_i) \cup \Opt_i$. By the exact
  same argument as \thmref{ptas:subset:dom}, $\locSol'$ is a
  dominating set. However, $\locSol'$ may not necessarily be
  connected.

  Let $\BVertices_i = \headsX{\bdDiv_i}$ be the set of head vertices
  of the boundary flowers of the $i$\th cluster.  Because the removed
  patch $\locSol_i$ is only connected to the rest of $\locSol$ via the
  boundary vertices $\BVertices_i$, each component of $\locSol$
  contains at least one boundary vertex in
  $\BVertices_i \cap \locSol$. Similarly, each component of
  $\optSet_i$ contains at least one boundary vertex in
  $\BVertices_i$. Together, every component of $\locSol'$ contains at
  least one vertex in $\BVertices_i$, so $\locSol'$ has at most
  $\cardin{\BVertices_i} \leq \exSize$ components.

  Consider the shortest path $\trA_{xy}$ within $\DomSet$ between any
  two vertices $x,y \in \locSol'$ that are in separate components of
  $\locSol'$. By minimality of $\trA_{xy}$, the interior vertices of
  $\trA_{xy}$ are not in $\locset'$.  If $\trA_{xy}$ has more than 4
  vertices, then there exists an intermediate vertex $v \in \trA_{xy}$
  that is adjacent to neither $x$ nor $y$.  Write
  $\trA_{xy} = \trA_{xv}|\trA_{vy}$, where $\trA_{xv}$ is the subpath
  from $x$ to $v$ and $\trA_{vy}$ is the subpath from $v$ to $y$. Both
  subpaths $\trA_{xv}$ and $\trA_{vy}$ contain at least two
  edges. Since $\demandX{v} \geq 1$, $v$ is adjacent to some vertex
  $z \in \locset'$. Since $x$ and $y$ lie in different in components,
  $z$ lies in a different component from either $x$ or $y$.  If $x$
  and $z$ lie in different components, then the path consisting of
  $\trA_{xv}$ followed by the edge from $v$ to $z$ is a shorter path
  than $\trA_{xy}$, a contradiction. A similar contradiction arises if
  $z$ and $y$ lies in different components. It follows, by
  contradiction, that $\trA_{xy}$ has at most 4 vertices, all of which
  lie in $\DomSet$.  By adding the entire path $\trA_{xy}$ to
  $\locSol'$, we can connect these two components by adding at most
  $2$ vertices from $D$.

  By repeatedly connecting the closest pair of components of
  $\locSol'$ like that, we can augment $\locSol'$ to a connected
  dominating set $\locSol''$ while adding at most
  $2 \cardin{\BVertices_i} \leq 2 \exSize$ vertices. If we expand our
  search size to $\exSize' = 3 \exSize$, then $\locSol''$ is a
  connected dominating set with
  $\cardin{\locSol'' \SetDiff \locSol} \leq \exSize'$, and the local
  optimality of $\locSol$ implies that
  \begin{math}
    \cardin{\locSol_i} %
    \leq %
    \cardin{\optset_i} + 2 \cardin{\BVertices_i}.
  \end{math}
  As in the previous proofs, summing this inequality over all $i$
  implies the claim.
\end{proof}

\lemref{ptas:subset:dom:3} extends to constantly bounded reach with an
added assumption.
\begin{lemma}
  \lemlab{ptas:subset:dom:4} %
  Let $\defGraph$ be a graph with $n$ vertices and with polynomial
  expansion of order $k$, and let $\DomSet \subseteq \Vertices$ be a
  given set. Assume that
  \begin{compactenum}[\quad(i)]
  \item for each vertex $v \in \Vertices$, there are associated demand
    $\demandX{v} \geq 1$ and reach $\reachX{v}$ constraints,
  \item $\MaxDemand = \max_v \demandX{v} = O(1)$ and
    $\MaxReach = \max_{v} \reachX{v} = O(1)$,
  \item the set $\DomSet$ is a valid dominating set complying with the
    demand and reach constraints,
  \item for any two vertices $u,v \in \DomSet$, the shortest path (in
    the number of edges) in $\graph$ between $u$ and $v$ is contained
    in $\GInduced{\DomSet}$.
  \end{compactenum}
  Then, for $\exSize = O\pth{\eps^{-(2k+6)}\log^{2k+5} (1/\eps)}$, the
  $\exSize$-local search algorithm computes, in $n^{O\pth{\exSize}}$
  time, a $(1 + \eps)$-approximation for the smallest cardinality
  subset of $\DomSet$ that is \emph{connected} and dominates
  $\Vertices$ under the reach and demand constraints.
\end{lemma}

\begin{proof}
  The same proof as that of \lemref{ptas:subset:dom:3} goes through,
  except now the shortest paths between distinct components can be
  shown to have length at most $2 (\MaxReach + 1)$ vertices. Condition
  (iv) is necessary to keep these paths lying in $D$. The search size
  is increased by a factor of $2 \MaxReach$ instead of 2, which is
  only a constant factor difference.
\end{proof}

\subsubsection{Discussion}

\begin{observation}[\PTAS for vertex cover for polynomial expansion
  graphs]
  \obslab{v:c:poly:expansion}%
  The algorithm of \thmref{g:hitting:set:cover} can be used to get a
  \PTAS for vertex cover. Indeed, let $\defGraph$ be an undirected
  graph with polynomial expansion. We introduce a new vertex in the
  middle of every edge of $\graph$, and let $\graphA$ be the resulting
  graph, with $\CovSet$ be the set of new vertices. Clearly, replacing
  an edge by a path of length two changes the expansion of a graph
  only slightly, see \defref{expansion}, and in particular, $\graphA$
  has polynomial expansion. Now, solving the dominating subset for
  $\CovSet$ as the set required covering, and $\Vertices$ as the
  initial dominating set, in the graph $\graphA$ solves the original
  vertex cover problem in the original graph. The desired \PTAS now
  follows from \thmref{g:hitting:set:cover}.
\end{observation}

\begin{remark}[\PTAS for graphs with subexponential expansion] %
  \remlab{s:exp:expansion}%
  As noted in \remref{w:div:small:excess}, one can still obtain
  $f(1/\eps)$-divisions for some (larger) function $f$ in graphs with
  hereditary separators size $O\pth{n/\log^{O(1)} n}$. To this end,
  one can verify (by the same proof as \thmref{p:e:separator}, see
  \cite{no-gcbe2-08,hq-naape-16-arxiv}) that for a small constant $c$,
  if a graph class $\class$ has expansion
  $\varphi(t) = O\pth{\exp(c' \cdot t^{c''})}$, for $c'$ and $c''$
  sufficiently small constants, then $\class$ has separators of the
  desired size $O\pth{n/\log^{O(1)} n}$. Thus, the above approximation
  algorithms yield a \PTAS (with much worse dependence on $\epsilon$)
  for any graph class $\class$ with subexponential expansion
  $\gradY{t}{\class} = O\pth{\exp(c' \cdot t^{c''})}$, where $c'$ and
  $c''$ are some constants. We are not aware of any natural graphs in
  this class that do not have polynomial expansion.
\end{remark}


\subsection{Geometric applications}
\seclab{geometric:applications}

The above implies \PTAS's for dominating set type problems on
low-density graphs.  Let $\ObjSet$ be a collection of objects in
$\Re^d$ and $\PntSet$ a collection of points.  Two natural geometric
optimization problems in this setting are:
\begin{compactenum}[\;\;(A)]
\item \emphi{Discrete hitting set}: Compute the minimum cardinality
  set $\PntSetA \subseteq \PntSet$ such that for every
  $\obj \in \ObjSet$, we have $\PntSetA \cap \obj \neq \emptyset$.
  That is, every object of $\ObjSet$ is stabbed by some point of
  $\PntSetA$.

  If we consider the natural intersection graph
  $\graph = \IGraph{\PntSet \cup \ObjSet}$ and the sets
  $\DomSet = \PntSet$ and $\CovSet = \ObjSet$, then this is an
  instance of dominating subset problem. The algorithm of
  \thmref{ptas:subset:dom} applies because $\graph$ is low density and
  therefore has polynomial expansion.

\item \emphi{Discrete set cover}: Compute the smallest cardinality set
  $\ObjSetA \subseteq \ObjSet$ such that for every point
  $\pnt \in \PntSet$, we have
  \begin{math}
    \pnt \in {\bigcup_{\obj \in \ObjSetA} \obj} .
  \end{math}
  That is, all the points of $\PntSet$ are covered by objects in
  $\ObjSetA$. Setting $\DomSet = \ObjSet$ and $\CovSet = \PntSet$
  (i.e., flipping the sets in the hitting set case), and arguing as
  above, implies a \PTAS.
\end{compactenum}
\smallskip%
For these geometric optimization problems, we can improve the running
time of \thmref{ptas:subset:dom} by applying the stronger separator
theorem for low-density graphs.
\begin{theorem}
  \thmlab{g:hitting:set:cover}%
  Let $\ObjSet$ be a collection of $m$ objects in $\Re^d$ with density
  $\cDensity$, $\PntSet$ be a set of $n$ points in $\Re^d$, and let
  $\eps > 0$ be a parameter. Then, for
  $\exSize = O\pth{\cDensity/\eps^{d}}$, the local search algorithm,
  with exchanges of size $\exSize$ implies the following:
  \begin{compactenum}[\quad(A)]
  \item An approximation algorithm that, in $O\pth{mn^{O(\exSize)}}$
    time, computes a set $\PntSetA \subseteq \PntSet$ that is an
    $(1+\eps)$-approximation for the smallest cardinality set that
    hits $\ObjSet$.

  \item An approximation algorithm that, in $O(n m^{O(\exSize)})$
    time, computes a set $\ObjSetA \subseteq \ObjSet$ that is an
    $(1+\eps)$-approximation for the smallest cardinality set that
    covers $\PntSet$.
  \end{compactenum}
\end{theorem}

\begin{proof}
  Since points have zero diameter, the union $\ObjSet \cup \PntSet$
  also has density $\cDensity+1$. This reduces geometric hitting set
  and discrete geometric set cover to dominating subset problem on the
  intersection graph of $\graph = \IGraph{\ObjSet \cup \PntSet}$.

  The approximation algorithm is described in \thmref{ptas:subset:dom}
  (applied to $\graph$). Here we can do slightly better, using smaller
  exchange size, as the graph $\graph$ has low density. To this end,
  observe that the analysis of \thmref{ptas:subset:dom} argues about
  the induced packing graph of $\graph$ for some $(2,1)$-shallow
  packing $\FamilyA$. By \lemref{shallow:cover:objects}, the graph
  $\graphA = \ICovGraph{\graph}{\FamilyA}$ has density
  $O(\cDensity 2^d ) = O( \cDensity)$.  Thus, by
  \corref{dev;p:e:l:dense} (B), $\graphA$ has a $\exSize$-division
  with excess $(\eps/4) \cardin{\VerticesX{\graphA}}$, where
  $\exSize = O\pth{ \cDensity /\eps^d }$.  The algorithm of
  \thmref{ptas:subset:dom} modified to use these improved divisions
  implies the result.
\end{proof}

\begin{remark}
  To our knowledge, the algorithms of \thmref{g:hitting:set:cover} are
  the first \PTAS's for discrete hitting set and discrete set cover
  with shallow fat triangles and similar fat objects. Previously, such
  algorithms were known only for disks and points in the plane.
\end{remark}


\section{Hardness of approximation}
\seclab{hardness}%

Some of the results of this section appeared in an unpublished
manuscript \cite{h-bffne-09}.  Chan and Grant \cite{cg-eaahr-14} also
prove some related hardness results, which were (to some extent) a
followup work to the aforementioned manuscript.

\subsection{A review of complexity terms} %
\seclab{complexity}

The \emphi{exponential time hypothesis} (\emphi{\ETH})
\cite{ip-ocks-01, ipz-wphse-01} is that \TrSAT can not be solved in
time better than $2^{\Omega(n)}$, where $n$ is the number of
variables. The \emphi{strong exponential time hypothesis}
(\emphi{\SETH}), is that the time to solve \kSAT is at least
$2^{c_k n}$, where $c_k$ converges to $1$ as $k$ increases.

A problem that is \APXHard does not have a \PTAS unless
$\POLYT = \NP$.  For example, it is known that \ProblemC{Vertex Cover}
is \APXHard even for a graph with maximum degree $3$
\cite{acgkm-ca-99}.  Thus, showing that a problem is \APXHard implies
that one can not do better than a constant
approximation. Specifically, if one can get a $(1+\eps)$-approximation
for such a problem, for any constant $\eps > 0$, then one can
$(1+\eps)$-approximate \TrSAT (for the max version of \TrSAT, the
purpose is to maximize the number of clauses satisfied). By the
\Term{PCP} Theorem, this would imply an exact algorithm for \TrSAT.

\begin{observation}%
  \obslab{eth:no:a}%
  Consider an instance of \TrSAT of size $m = c \log^2 n$, for some
  constant $c > 0$ sufficiently large. \ETH implies that we cannot
  solve this instance in polynomial time, since the running time
  required to solve this instance is $2^{\Omega(m)}$, which is super
  polynomial. (This argument works for any function $f(n)$ such that
  $\log n = o(f(n))$.)

  This innocuous observation has a surprising implication -- we cannot
  even $(1-\eps)$ approximate a solution for such an instance by the
  \Term{PCP} result. Namely, \ETH implies that even polylogarithmic
  sized instances cannot be solved in polynomial time.
\end{observation}%

\begin{observation}
  Showing that a problem $X$ is \APXHard implies that:
  \begin{compactenum}[(A)]
  \item The problem $X$ does not have a \PTAS (unless $\POLYT=\NP$).

  \item Under \ETH, the problem $X$ does not have a \QPTAS, where a
    \QPTAS is an $(1+\eps)$-approximation algorithm with running time
    $n^{O(\poly( \log n,1/\eps))}$.

  \item Furthermore, under \ETH, polylogarithmic sized instances of
    $X$ cannot be approximated to within a $(1\pm\eps)$-multiplicative
    factor in polynomial time.
  \end{compactenum}
\end{observation}

\subsection{Discrete hitting set for fat triangles}

In the \emphi{fat-triangles discrete hitting set problem}, we are
given a set of points in the plane $\PntSet$ and a set of fat
triangles $\TriSet$, and want to find the smallest subset of $\PntSet$
such that each triangle in $\TriSet$ contains at least one point in
the set.
\begin{figure}
  \centerline{%
    \begin{tabular}{c|c|c|c}
      \IncludeGraphics[page=1,width=.21\textwidth]{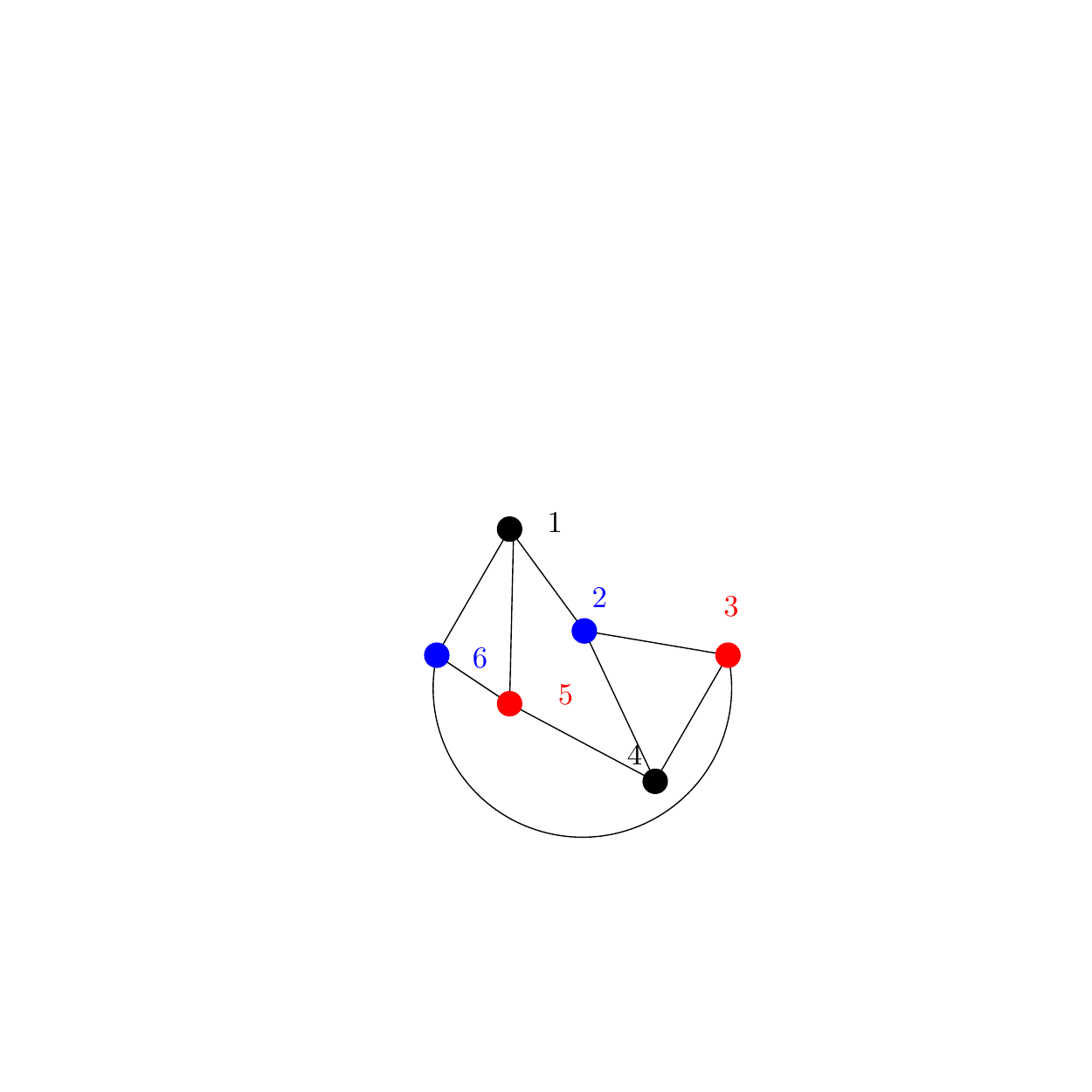}%
      &%
        \quad\IncludeGraphics[page=4,width=.21\textwidth]{figs/triangles}%
      &
        \quad\IncludeGraphics[page=2,width=.21\textwidth]{figs/triangles}%
      &%
        \quad%
        \IncludeGraphics[page=3,width=.21\textwidth]{figs/triangles}\\
      (A) & (B) & (C) & (D)
    \end{tabular}%
  }%
  \caption{%
    Illustration of the proof of \lemref{no:PTAS:fat:hit:set}: %
    (A) A $3$-regular graph with its $3$ coloring. %
    (B) Placing the vertices on a circle. %
    (C) Three edges and their associated triangles. %
    (D) All the triangles.
  }%
  \figlab{illustration}
\end{figure}
\begin{lemma}
  \lemlab{no:PTAS:fat:hit:set}%
  There is no \PTAS for the fat-triangle discrete hitting set problem,
  unless $\POLYT=\NP$.
  One can prespecify an arbitrary constant $\delta > 0$, and the claim
  would hold true even if the following conditions hold on the given
  instance $(\PntSet,\TriSet)$:
  \begin{compactenum}[\quad(A)]
  \item Every angle of every triangle in $\TriSet$ is between
    $60-\delta$ and $60 + \delta$ degrees.

  \item No point of $\PntSet$ is covered by more than three triangles
    of $\TriSet$.

  \item The points of $\PntSet$ are in convex position.

  \item All the triangles of $\TriSet$ are of similar
    size. Specifically, each triangle has side length in the range
    (say) $(\sqrt{3}-\delta, \sqrt{3}+\delta)$.

  \item The points of $\PntSet$ are a subset of the vertices of the
    triangles of $\TriSet$.
  \end{compactenum}
\end{lemma}

\begin{proof}
  Let $\graph = (\Vertices,\Edges)$ be a connected instance of
  \ProblemC{Vertex Cover} which has maximum degree three, and it is
  not an odd cycle. We remind the reader that \ProblemC{Vertex Cover}
  is \APXHard for such instances \cite{acgkm-ca-99}.

  By Brook's theorem \cite{cr-btb-15}\footnote{Brook's theorem states
    that any connected undirected graph $\graph$ with maximum degree
    $\Delta$, the chromatic number of $\graph$ is at most $\Delta$
    unless $\graph$ is a complete graph or an odd cycle, in which case
    the chromatic number is $\Delta+1$.}, this graph is three
  colorable, and let $\Vertices_1, \Vertices_2, \Vertices_3$ be the
  partition of $\Vertices$ by their colors. Let
  $\pnt_1, \pnt_2, \pnt_3$ be three points on the unit circle that
  form a regular triangle. For $i=1,2,3$, place a circular interval
  $\Interval_i$ centered at $\pnt_i$ of length $\delta/100$. Now, for
  $i=1,2,3$, we place the vertices of $\Vertices_i$ as distinct points
  in $\Interval_i$.

  Let $\PntSetA_0 = \Vertices$ and $m = \cardin{\EdgesX{\graph}}$.
  For $i=1, \ldots, m$, let $u_i v_i$ be the $i$\th edge of
  $\graph$. Assume, for the sake of simplicity of exposition, that
  $u_i \in \Vertices_1$ and $v_i \in \Vertices_2$.  Pick an arbitrary
  point $\pntA_i $ in $\Interval_3 \setminus \PntSetA_{i-1}$, and
  create the triangle $T_i = \triangle u_i v_i \pntA_i$. Set
  $\PntSetA_i = \PntSetA_{i-1} \cup \brc{\pntA_i}$, and continue to
  the next edge.

  At the end of this process, we have $m$ triangles
  $\TriSet = \brc{ T_1, \ldots, T_m}$ that are arbitrarily close to
  being regular triangles, and all their edges are arbitrarily close
  to being of the same length, see \figref{illustration}. It is easy
  to verify that a minimum cardinality set of points
  $U \subseteq \Vertices$ that hits all the triangles in $\TriSet$ is
  a minimum vertex cover of $\graph$.
\end{proof}

\subsection{Friendly geometric set cover}

Let $\PntSet$ be a set of $n$ points in the plane, and $\Family$ be a
set of $m$ regions in the plane, such that
\begin{compactenum}[\qquad(I)]
\item the shapes of $\Family$ are convex, fat, and of similar size,
\item the boundaries of any pair of shapes of $\Family$ intersect in
  at most $6$ points,
\item the union complexity of any $m$ shapes of $\Family$ is $O(m)$,
  and
\item any point of $\PntSet$ is covered by a constant number of shapes
  of $\Family$.  %
\end{compactenum}%
\smallskip%
We are interested in finding the minimum number of shapes of $\Family$
that covers all the points of $\PntSet$.  This variant is the
\emphi{friendly geometric set cover} problem.

\begin{lemma}
  \lemlab{no:PTAS:friendly:s:c}%
  There is no \PTAS for the friendly geometric set cover problem,
  unless $\POLYT=\NP$.
\end{lemma}

\begin{proof}
  Let $U$ be a set of $n$ elements, and $\Family$ a set of subsets of
  $U$ each containing at most $k$ elements of $U$. In the
  \emphi{minimum $k$-set cover} problem, we want to find the smallest
  subcollection $\FamilyA \subseteq \Family$ that covers $U$. The
  problem is \MaxSNPHard for $k \geq 3$, meaning there is no \PTAS
  unless $\POLYT = \NP$ \cite{acgkm-ca-99}.

  We will reduce an instance $(U, \Family)$ of the minimum $k$-set
  cover problem (for $k=3$) into an instance of the friendly geometric
  set cover problem.

  \begin{figure}
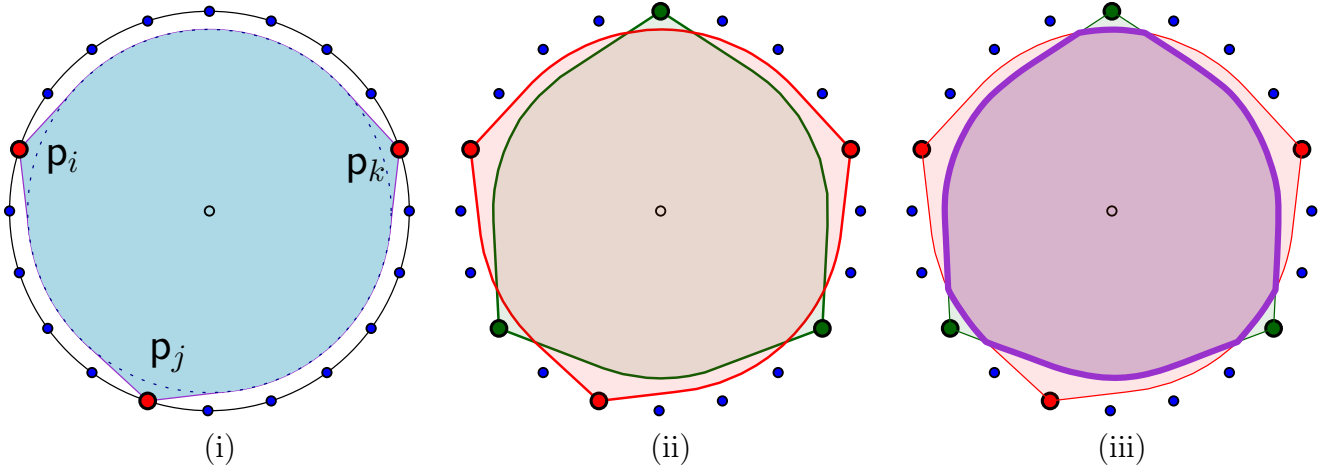

    \centerline{%
      \begin{tabular}{c%
        cc}
        \begin{minipage}[b]{0.3\linewidth}
          \centerline{\IncGraphPage%
            [width=0.99\linewidth]%
            {figs}{gear}{1}}
        \end{minipage}%
        &%
          \begin{minipage}[b]{0.3\linewidth}
            \centerline{\IncGraphPage%
              [width=0.99\linewidth]%
              {figs}{gear}{2}}
          \end{minipage}%
        &%
          \begin{minipage}[b]{0.3\linewidth}
            \centerline{\IncGraphPage[width=0.99\linewidth]%
              {figs}{gear}{3}}
          \end{minipage}
        \\
        \smallskip (i) & (ii) & (iii)
      \end{tabular}%
    }

    \caption{(i) A region $\objA$ constructed for the set
      $S_t = \brc{u_i, u_j, u_k}$. Observe that in the
      construction, the inner disk is even bigger. As such, no
      two points are connected by an edge of the convex-hull when
      we add in the inner disk to the convex-hull. As such, each
      point ``contribution'' to the region $\objA$ is separated
      from the contribution of other points. %
      (ii) How  two such
      regions together.
      (iii) Their intersection.  }
    \figlab{intersection:of:regions}
  \end{figure}

  Let $U = \brc{u_1,\ldots, u_n}$ be a set of $n$ elements, and
  $\Family = \brc{S_1, \ldots, S_m}$ a collection of $m$ subsets of
  $U$. We place $n$ points equally spaced on the unit radius circle
  centered at the origin, and let
  $\PntSet = \brc{\pnt_1, \ldots, \pnt_n}$ be the resulting set of
  points. For each point $u_i \in U$, let $f(u_i) = \pnt_i$.  For each
  set $S_i \in \Family$ (of size at most 3), we define the region
  \begin{align*}
    \objA_i%
    =%
    \CH\pth{ \; \DiskOrg\pth{ 1-\frac{i}{10n^2m}} \; \cup \; f(S_i)
    \;},
  \end{align*}
  where $\CH$ is the convex hull,
  $f(S_i) = \cup_{x \in S_i} \brc{ f(x)}$, and $\DiskOrg( r )$ denotes
  the disk of radius $r$ centered at the origin. Visually, $\objA_i$
  is a disk with three (since $k=3$) teeth coming out of it, see
  \figref{intersection:of:regions}. Note that the boundary of two such
  shapes intersects in at most $6$ points.

  It is now easy to verify that the resulting instance of geometric
  set cover $\pth{\PntSet, \brc{\objA_1, \ldots, \objA_m}}$ is
  friendly, and clearly any cover of $\PntSet$ by these shapes can be
  interpreted as a cover of $U$ by the corresponding sets of
  $\Family$. Thus, a \PTAS for the friendly geometric set cover
  problem implies a \PTAS for the minimum $k$-set cover, which is
  impossible unless $\POLYT=\NP$.
\end{proof}

\subsection{Set cover by fat triangles}
\seclab{hardness-set-cover-fat-triangles}

In the \emphi{fat-triangle set cover problem}, specified by a set of
points in the plane $\PntSet$ and a set of fat triangles $\TriSet$,
one wants to find the minimum subset of $\TriSet$ such that its union
covers all the points of $\PntSet$.

\begin{lemma}
  \lemlab{no:PTAS:fat:tr:set:cover}%
  There is no \PTAS for the fat-triangle set cover problem, unless
  $\POLYT=\NP$.  Furthermore, one can prespecify an arbitrary constant
  $\delta > 0$, and the claim would hold true even if the following
  conditions hold on the given instance $(\PntSet,\TriSet)$:
  \begin{compactenum}[\quad(A)]
  \item The minimum angle of all the triangles of $\TriSet$ is larger
    than $45-\delta$ degrees.

  \item No point of $\PntSet$ is covered by more than two triangles of
    $\TriSet$.

  \item The points of $\PntSet$ are in convex position.

  \item All the triangles of $\TriSet$ are of similar
    size. Specifically, each triangle has diameter in the range (say)
    $(2-\delta, 2]$.

  \item Each triangle of $\TriSet$ has two angles in the range
    $(45-\delta, 45+\delta)$, and one angle in the range
    $(90-\delta, 90+\delta)$.

  \item The vertices of the triangles of $\TriSet$ are the points of
    $\PntSet$.
  \end{compactenum}
\end{lemma}

\begin{proof}
  Consider a graph $\graph$ with maximum degree three, and observe
  that a \ProblemC{Vertex Cover} problem in such a graph can be
  reduced to \ProblemC{Set Cover} where every set is of size at most
  $3$. Indeed, the ground set $U$ is the edges of $\graph$, and every
  vertex $v \in \VerticesX{\graph}$ gives a rise to the set
  $S_v = \Set{uv \in \EdgesX{\graph}}{ u \in \VerticesX{\graph}}$,
  which is of size at most $3$. Clearly, any cover $C$ of size $t$ for
  the set system
  \begin{math}
    \mathcal{X} = \pth{\Bigl. U, \, \brc{S_v \sepw{v \in
          \VerticesX{\graph} \bigr.}}},
  \end{math}
  has a corresponding vertex cover of $\graph$ of the same size. Thus,
  \ProblemC{Set Cover} with every set of size (at most) three is
  \APXHard (this is of course well known). Note that in this set cover
  instance, every element participates in exactly two sets (i.e., the
  two vertices adjacent to the original edge).

  The graph $\graph$ has maximum degree three, and by Vizing's theorem
  \cite{bm-gta-76}, it is $4$ edge-colorable\footnote{Vizing's theorem
    states that a graph with maximum degree $\Delta$ can be edge
    colored by $\Delta+1$ colors. In this specific case, one can reach
    the same conclusion directly from Brook's theorem. Indeed, in our
    case, the adjacency graph of the edges has degree at most $4$, and
    it does not contain a clique of size $4$. As such, this graph is
    $4$-colorable, implying the original graph edges are
    $4$-colorable.}.  With regards to the set problem, the ground set
  of the set system $\mathcal{X}$ can be colored by $4$ colors such
  that no set in this set system has a color appearing more than once.

  We are given an instance of the \ProblemC{Vertex Cover} problem for
  a graph with maximum degree $3$, and we transform it into a set
  cover instance as mentioned above, denoted by
  $\mathcal{X} = \pth{ U, \Family_{\mathcal{X}}}$. Let
  $n = \cardin{U}$, and color $U$ (as described above) by $4$ colors
  such that no set of $\mathcal{X}$ has the same color repeated twice,
  let $U_1, \ldots, U_4$ be the partition of $U$ by the color of the
  points.

  \parpic[r]{\IncludeGraphics{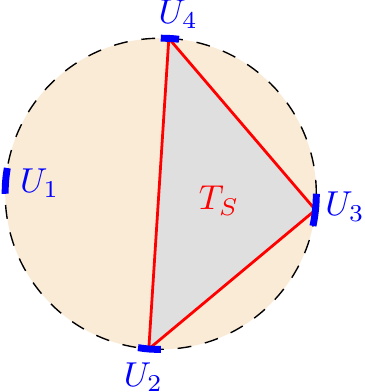}}
  Let $\mathcal{C}$ denote the circle of radius one centered at the
  origin. We place four relatively short arcs on $C$, placed on the
  four intersection points of $C$ with the $x$ and $y$ axes, see
  figure on the right.  Let $I_1, \ldots, I_4$ denote these four
  circular intervals.  We equally space the elements of $U_i$ (as
  points) on the interval $I_{i}$, for $i=1,\ldots, 4$. Let $\PntSet$
  be the resulting set of points.

  For every set $S \in \Family_{\mathcal{X}}$, take the convex hull of
  the points corresponding to its elements as its representing
  triangle $T_S$. Note, that since the vertices of $T_S$ lie on three
  intervals out of $I_1, I_2, I_3, I_4$, it follows that it must be
  fat, for all $S \in \Family_{\mathcal{X}}$. As such, the resulting
  set of triangles
  $\TriSet = \brc{T_S \sepw{S \in \Family_{\mathcal{X}}}}$ is fat, and
  clearly there is a cover of $\PntSet$ by $t$ triangles of $\TriSet$
  if and only if the original set cover problem has a cover of size
  $t$.

  Any triangle having its three vertices on three different intervals
  of $I_1, \ldots, I_4$ is close to being an isosceles triangle with
  the middle angle being $90$ degrees. As such, by choosing these
  intervals to be sufficiently short, any triangle of $\TriSet$ would
  have a minimum degree larger than, say, $45-\delta$ degrees, and
  with diameter in the range between $2-\delta$ and $2$.

  This is clearly an instance of the fat-triangle set cover
  problem. Solving it is equivalent to solving the original
  \ProblemC{Vertex Cover} problem, but since it is \APXHard, it
  follows that the fat-triangle set cover problem is \APXHard.
\end{proof}

\begin{remark}
  For fat triangles of similar size a constant factor approximation
  algorithm is known \cite{cv-iaags-07}.
  \lemref{no:PTAS:fat:tr:set:cover} implies that one can do no
  better. Naturally, it might be possible to slightly improve the
  constant of approximation provided by the algorithm of Clarkson and
  Varadarajan \cite{cv-iaags-07}.
  However, for fat triangles of different sizes, only a $\log^*$
  approximation is known \cite{abes-ibulf-14}. It is natural to ask if
  this can be improved.
\end{remark}

\subsubsection{Extensions}

\begin{lemma}
  \lemlab{no:P:T:A:S:circles}%
  Given a set of points $\PntSet$ in the plane and a set of circles
  $\Family$, finding the minimum number of circles of $\Family$ that
  covers $\PntSet$ is \APXHard; that is, there is no \PTAS for this
  problem.
\end{lemma}
\begin{proof}
  Slightly perturb the point set used in the proof of
  \lemref{no:PTAS:fat:tr:set:cover}, so that no four points of it are
  co-circular. Let $\PntSet$ denote the resulting set of points. For
  every set $S \in \Family_{\mathcal{X}}$, we now take the circle
  passing through the three corresponding points. Clearly, this
  results in a set of circles (that are almost identical, but yet all
  different), such that finding the minimum number of circles covering
  the set $\PntSet$ is equivalent to solving the original problem.
\end{proof}

\begin{lemma}
  \lemlab{no:PTAS:cover:planes}%
  Given a set of points $\PntSetA$ in $\Re^3$ and a set of planes
  $\Family$, finding the minimum number of planes of $\Family$ that
  covers $\PntSetA$ is \APXHard; that is, there is no \PTAS for this
  problem.
\end{lemma}
\begin{proof}
  Let $\PntSet$ be the point set and $\Family$ be the set of circles
  constructed in the proof of \lemref{no:P:T:A:S:circles}, and map
  every point in it to three dimensions using the mapping
  $f: (x,y) \rightarrow (x,y,x^2 + y^2)$.  This is a standard lifting
  map used in computing planar Delaunay triangulations via convex-hull
  in three dimensions, see \cite{bcko-cgaa-08}.  Let
  $\PntSetA = f(\PntSet)$ be the resulting point set.

  It is easy to verify that a circle of $c \in \Family$ is mapped by
  $f$ into a curve that lies on a plane. We will abuse notations
  slightly, and use $f(c)$ to denote this plane.  Let
  $\FamilyA = f(\Family)$. Furthermore, for a circle $c \in \Family$,
  we have that $f(c \cap \PntSet) = f(c) \cap \PntSetA$.  Namely,
  solving the set cover problem $(\PntSetA, \FamilyA)$ is equivalent
  to solving the original set cover instance $(\PntSet, \Family)$.
\end{proof}

The recent work of Mustafa \etal \cite{mr-irghs-10} gave a \QPTAS for
set cover of points by disks (i.e., circles with their interior), and
for set cover of points by half-spaces in three dimensions. Thus,
somewhat surprisingly, the ``shelled'' version of these problems are
harder than the filled-in version.  %

\subsection{Independent set of triangles in 3D}

Given a set $\ObjSet$ of $n$ objects in $\Re^d$ (say, triangles in
3d), we are interested in computing a maximum number of objects that
are \emphi{independent}; that is, no pair of objects in this set
(i.e., independent set) intersects. This is the geometric realization
of the \emphi{independent set} problem for the intersection graph
induced by these objects.

\begin{lemma}
  \lemlab{no:PTAS:3:d:i:s}%
  There is no \PTAS for the maximum independent set of triangles in
  $\Re^3$, unless $\POLYT=\NP$.
\end{lemma}
\begin{proof}
  The problem \ProblemC{Independent Set} is \APXHard even for graphs
  with maximum degree $3$ \cite{acgkm-ca-99}. Let
  $\graph=(\Vertices,\Edges)$ be a given graph with maximum degree
  $3$, where $\Vertices = \brc{v_1,\ldots, v_n}$.  We will create a
  set of triangles, such that their intersection graph is $\graph$.

  If one spreads $n$ points $\pnt_1,\ldots, \pnt_n$ on the positive
  branch of the moment curve in $\Re^3$ \cite{s-eubnf-91,
    ek-alnfc-03}, their Voronoi diagram is \emphi{neighborly}; that
  is, every Voronoi cell is a convex polytope that shares a non-empty
  two dimensional boundary face with each of the other cells of the
  diagram. Let $C_i$ denote the cell of the point $\pnt_i$ in this
  Voronoi diagram, for $i=1,\ldots, n$.

  Now, for every vertex $v_i \in \Vertices$, we form a set $\PntSet_i$
  of (at most) three points, as follows. If $v_iv_j \in \Edges$, then
  we place a point $p_{ij}$ on the common boundary of $C_i$ and $C_j$,
  and we add this point to both $\PntSet_i$ and $\PntSet_j$.  After
  processing all the edges in $\Edges$, each point set $\PntSet_i$ has
  at most three points, as the maximum degree in $\graph$ is three.

  For $i=1,\ldots, n$, let $f_i$ be the triangle formed by the
  convex-hull of $\PntSet_i$ (if $\PntSet_i$ has fewer than three
  points then the triangle is degenerate).

  Let $\TriSet = \brc{f_1,\ldots, f_n}$.  Observe that the triangles
  of $\TriSet$ are disjoint except maybe in their common vertices, as
  their interior is contained inside the interior of $C_i$, and the
  cells $C_1, \ldots, C_n$ are interior disjoint. Clearly
  $f_i \cap f_j \ne \emptyset$ if and only if $v_iv_j \in E$. Thus,
  finding an independent set in $\graph$ is equivalent to finding an
  independent set of triangles of the same size in $\TriSet$. We
  conclude that the problem of finding maximum independent set of
  triangles is \APXHard, and as such does not have a \PTAS unless
  $\POLYT=\NP$.
\end{proof}

Implicit in the above proof is that any graph can be realized as the
intersection graph of convex bodies in $\Re^3$ (we were a bit more
elaborate for the sake of completeness and since we needed slightly
more structure). This is well known and can be traced to a result of
Tietze from 1905 \cite{t-upnr-05}.

\subsection{Hardness of approximation with respect %
  to depth}


Using \obsref{eth:no:a}, we can derive hardness of approximation
results even for ``small'' instances.  Here, we reconsider the
\emph{geometric set cover problem}: Given a set of objects $\ObjSet$
in $\Re^d$, and a set of points $\PntSet$, we would like to find
minimum cardinality subset of the objects in $\ObjSet$ that covers the
points of $\PntSet$.

\begin{lemma}
  \lemlab{e:t:h:g:polylog:d}%
  Assuming the exponential time hypothesis ($\ETH$) (see
  \secref{complexity}), consider a given set of fat triangles
  $\TriSet$ of density $\cDensity$, and a set of points $\PntSet$,
  such that $\cardin{\PntSet} + \cardin{\TriSet} = O(n)$.  We have the
  following:
  \begin{compactenum}[\quad(A)]
  \item If $\cDensity = \Omega(\log^c n)$ then one cannot
    $(1+\eps)$-approximate the geometric set cover (or geometric
    hitting set) instance $(\PntSet,\TriSet)$ in polynomial time,
    where $c$ is a sufficiently large constant.

  \item There is an absolute constant $c''$, such that no
    $(1+\eps)$-approximation algorithm for the geometric set cover (or
    hitting set) instance $(\PntSet,\TriSet)$ has running time
    $n^{ \cDensity^{c''} /\eps^{O(1)}}$.
  \end{compactenum}
\end{lemma}
\begin{proof}
  (A) Suppose we had such a \PTAS, and consider an instance
  $\Instance$ of \TrSAT of size at least $c' \log^2 n$, where $c'$ is
  a sufficiently large constant.  \ETH implies that any algorithm
  solving such an instance must have running time at least
  $n^{\Omega(\log n)}$.  On the other hand, the instance $\Instance$
  can be converted to a set cover instance of fat triangles with
  $\polylog n$ triangles/points and $\polylog n$ density, by
  \lemref{no:PTAS:fat:tr:set:cover}. As such, a \PTAS in this case,
  would contradict \ETH.

  (B) Consider a constant $c''$ sufficient small.  By part (A), an
  instance of \TrSAT with $\log^{2} n$ variables, can be converted
  into an instance of geometric set cover (of fat triangles) with
  depth $\log^{c_3} n$ (where $c_3$ is some constant). But then, if
  $c'' <1/c_3$, this implies that this instance be solved in
  polynomial time, contradicting \ETH.

  The same conclusions holds for geometric hitting set, by using
  \lemref{no:PTAS:fat:hit:set}.
\end{proof}

\section{Conclusions}
\seclab{conclusions}

In this paper, we studied the class of graphs arising out of low
density objects in $\Re^d$, and showed that they are subclass of
graphs with polynomial expansion. We provided \PTAS's for independent
set and dominating set problems (and some variants) for such
graphs. This gives rise to \PTAS for some generic variants of these
problems. Coupled with hardness results, we characterize the
complexity of geometric variants of set cover and hitting set as a
function of depth (for example, for fat triangles).

At this point in time, it seems interesting to better understand low
density graphs. In particular, how exactly do they relate to graphs of
low genus, and whether one can develop efficient approximation
algorithms and hardness of approximations to other problems for this
family of graphs. (In particular, there is strong evidence that low
genus graphs are low density graphs.) For example, as a concrete
problem, can one get a \PTAS for TSP for low-density graphs or
polynomial expansion graphs?

\paragraph{Acknowledgments.}

The authors thank Mark \si{de} Berg for useful discussions related to
the problems studied in this paper. In particular, he pointed out the
improved bound for \lemref{separator:low:density}. We also thank the
anonymous referees. We are particularly grateful to the anonymous
referee of a previous version of this paper who pointed out the
connection of our work to graphs with bounded expansion.

\hypersetup{%
  allcolors=black%
}

\BibTexMode{%
\newcommand{\etalchar}[1]{$^{#1}$}
 \providecommand{\CNFX}[1]{ {\em{\textrm{(#1)}}}}
  \providecommand{\tildegen}{{\protect\raisebox{-0.1cm}{\symbol{'176}\hspace{-0.03cm}}}}
  \providecommand{\SarielWWWPapersAddr}{http://sarielhp.org/p/}
  \providecommand{\SarielWWWPapers}{http://sarielhp.org/p/}
  \providecommand{\urlSarielPaper}[1]{\href{\SarielWWWPapersAddr/#1}{\SarielWWWPapers{}/#1}}
  \providecommand{\Badoiu}{B\u{a}doiu}
  \providecommand{\Barany}{B{\'a}r{\'a}ny}
  \providecommand{\Bronimman}{Br{\"o}nnimann}  \providecommand{\Erdos}{Erd{\H
  o}s}  \providecommand{\Gartner}{G{\"a}rtner}
  \providecommand{\Matousek}{Matou{\v s}ek}
  \providecommand{\Merigot}{M{\'{}e}rigot}
  \providecommand{\Hastad}{H\r{a}stad\xspace}
  \providecommand{\CNFCCCG}{\CNFX{CCCG}}
  \providecommand{\CNFBROADNETS}{\CNFX{BROADNETS}}
  \providecommand{\CNFESA}{\CNFX{ESA}}
  \providecommand{\CNFFSTTCS}{\CNFX{FSTTCS}}
  \providecommand{\CNFIJCAI}{\CNFX{IJCAI}}
  \providecommand{\CNFINFOCOM}{\CNFX{INFOCOM}}
  \providecommand{\CNFIPCO}{\CNFX{IPCO}}
  \providecommand{\CNFISAAC}{\CNFX{ISAAC}}
  \providecommand{\CNFLICS}{\CNFX{LICS}}
  \providecommand{\CNFPODS}{\CNFX{PODS}}
  \providecommand{\CNFSWAT}{\CNFX{SWAT}}
  \providecommand{\CNFWADS}{\CNFX{WADS}}

}

\BibLatexMode{\printbibliography}

\begin{thebibliography}{vdSOdBV98}

\bibitem[ACG{\etalchar{+}}99]{acgkm-ca-99}
G.~Ausiello, P.~Crescenzi, G.~Gambosi, V.~Kann, A.~Marchetti-Spaccamela, and
  M.~Protasi.
\newblock  {\em Complexity and approximation}.
\newblock Springer-Verlag, Berlin, 1999.

\bibitem[AdBES14]{abes-ibulf-14}
\href{http://cis.poly.edu/~aronov/}{B.~{Aronov}}, \href{http://www.win.tue.nl/~mdberg/}{M.~de~{Berg}}, E.~Ezra, and \href{http://www.math.tau.ac.il/~michas}{M.~{Sharir}}.
\newblock \href{http://dx.doi.org/10.1137/120891241}{Improved bounds for the
  union of locally fat objects in the plane}.
\newblock {\em SIAM J. Comput.}, 43(2):543--572, 2014.

\bibitem[AES10]{aes-ssena-10}
\href{http://cis.poly.edu/~aronov/}{B.~{Aronov}}, E.~Ezra, and \href{http://www.math.tau.ac.il/~michas}{M.~{Sharir}}.
\newblock \href{http://dx.doi.org/10.1137/090762968}{Small-size
  {$\varepsilon$}-nets for axis-parallel rectangles and boxes}.
\newblock {\em SIAM J. Comput.}, 39(7):3248--3282, 2010.

\bibitem[And70]{a-ocpls-70}
E.M. Andreev.
\newblock  On convex polyhedra in lobachevsky spaces.
\newblock {\em Sbornik: Mathematics}, 10:413--440, April 1970.

\bibitem[APS08]{aps-sugo-08}
\href{http://www.cs.duke.edu/~pankaj}{P.~K.~{Agarwal}}, \href{http://www.math.nyu.edu/~pach}{J.~{Pach}}, and \href{http://www.math.tau.ac.il/~michas}{M.~{Sharir}}.
\newblock
  \href{https://users.cs.duke.edu/~pankaj/publications/surveys/union.pdf}{State
  of the union--of geometric objects}.
\newblock In J.~E. Goodman, \href{http://www.math.nyu.edu/~pach}{J.~{Pach}}, and R.~Pollack, editors, {\em Surveys in
  Discrete and Computational Geometry Twenty Years Later}, volume 453 of {\em
  Contemporary Mathematics}, pages 9--48. Amer. Math. Soc., 2008.

\bibitem[AW13]{aw-asmwi-13}
A.~Adamaszek and A.~Wiese.
\newblock  Approximation schemes for maximum weight independent set of
  rectangles.
\newblock In {\em Proc. 54th Annu. IEEE Sympos. Found. Comput. Sci.
  {\em(FOCS)}}, pages 400--409, 2013.

\bibitem[AW14]{aw-qmwis-14}
A.~Adamaszek and A.~Wiese.
\newblock  A {QPTAS} for maximum weight independent set of polygons with
  polylogarithmic many vertices.
\newblock In {\em Proc. 25th ACM-SIAM Sympos. Discrete Algs. {\em(SODA)}},
  pages 400--409, 2014.

\bibitem[Bak94]{b-aancp-94}
B.~S. Baker.
\newblock  Approximation algorithms for {NP}-complete problems on planar
  graphs.
\newblock {\em \href{http://www.acm.org/jacm/}{J. Assoc. Comput. {Mach.}}}, 41:153--180, 1994.

\bibitem[BM76]{bm-gta-76}
J.~A. Bondy and U.~S.~R. Murty.
\newblock \href{http://www.ecp6.jussieu.fr/pageperso/bondy/books/gtwa/gtwa\
  .html}{{\em Graph Theory with Applications}}.
\newblock North-Holland, 1976.

\bibitem[CCH12]{cch-smcpg-12}
C.~Chekuri, K.~Clarkson, and \href{http://sarielhp.org}{S.~{{Har-Peled}}}.
\newblock \href{http://sarielhp.org/papers/08/multi_cover}{On the set
  multi-cover problem in geometric settings}.
\newblock {\em ACM Trans. Algo.}, 9(1):9, 2012.

\bibitem[CG09]{cg-epgig-09}
J.~Chalopin and D.~Gon{\c{c}}alves.
\newblock  Every planar graph is the intersection graph of segments in the
  plane: extended abstract.
\newblock In {\em Proc. 41st Annu. ACM Sympos. Theory Comput. {\em(STOC)}},
  pages 631--638, 2009.

\bibitem[CG14a]{cg-spfgem-14}
S.~Cabello and D.~Gajser.
\newblock \href{http://arxiv.org/abs/1410.5778}{Simple ptas's for families of
  graphs excluding a minor}.
\newblock {\em CoRR}, abs/1410.5778, 2014.

\bibitem[CG14b]{cg-eaahr-14}
Timothy~M. Chan and Elyot Grant.
\newblock  Exact algorithms and {APX}-hardness results for geometric packing
  and covering problems.
\newblock {\em Comput. Geom. Theory Appl.}, 47(2):112--124, 2014.

\bibitem[CH12]{ch-aamis-12}
\href{http://www.math.uwaterloo.ca/~tmchan/}{T.~M.~{Chan}} and \href{http://sarielhp.org}{S.~{{Har-Peled}}}.
\newblock  Approximation algorithms for maximum independent set of
  pseudo-disks.
\newblock {\em \href{http://link.springer.com/journal/454}{Discrete Comput. {}Geom.}}, 48:373--392, 2012.

\bibitem[Cha03]{c-ptasp-03}
Timothy~M. Chan.
\newblock  Polynomial-time approximation schemes for packing and piercing fat
  objects.
\newblock {\em J. Algorithms}, 46(2):178--189, 2003.

\bibitem[CR15]{cr-btb-15}
Daniel~W. Cranston and Landon Rabern.
\newblock  Brooks' theorem and beyond.
\newblock {\em J. Graph Theo.}, 80(3):199--225, 2015.

\bibitem[CV07]{cv-iaags-07}
\href{http://cm.bell-labs.com/who/clarkson/}{K.~L. {Clarkson}} and \href{http://www.cs.uiowa.edu/~kvaradar/}{K.~R. {Varadarajan}}.
\newblock  Improved approximation algorithms for geometric set cover.
\newblock {\em \href{http://link.springer.com/journal/454}{Discrete Comput. {}Geom.}}, 37(1):43--58, 2007.

\bibitem[dB08]{b-ibucf-08}
\href{http://www.win.tue.nl/~mdberg/}{M.~de~{Berg}}.
\newblock  Improved bounds on the union complexity of fat objects.
\newblock {\em \href{http://link.springer.com/journal/454}{Discrete Comput. {}Geom.}}, 40(1):127--140, 2008.

\bibitem[dBCKO08]{bcko-cgaa-08}
\href{http://www.win.tue.nl/~mdberg/}{M.~de~{Berg}}, \href{http://www.win.tue.nl/~ocheong}{O.~{Cheong}}, {M. van} Kreveld, and \href{http://www.cs.uu.nl/people/markov/}{M.~H. {Overmars}}.
\newblock \href{http://www.cs.uu.nl/geobook/}{{\em Computational Geometry:
  Algorithms and Applications}}.
\newblock Springer-Verlag, Santa Clara, CA, USA, 3rd edition, 2008.

\bibitem[dBKSV02]{bksv-rimga-02}
\href{http://www.win.tue.nl/~mdberg/}{M.~de~{Berg}}, M.~J. Katz, {A. F.}~{van~der} Stappen, and J.~Vleugels.
\newblock  Realistic input models for geometric algorithms.
\newblock {\em Algorithmica}, 34(1):81--97, 2002.

\bibitem[DN15]{dn-ssspe-15}
Z.~{Dvo{\v{r}}{\'{a}}k} and S.~{Norin}.
\newblock \href{http://arxiv.org/abs/1504.04821}{{Strongly sublinear separators
  and polynomial expansion}}.
\newblock {\em ArXiv e-prints}, April 2015.

\bibitem[EK03]{ek-alnfc-03}
\href{http://compgeom.cs.uiuc.edu/~jeffe/}{J.~{Erickson}} and S.~Kim.
\newblock  Arbitrarily large neighborly families of congruent symmetric convex
  3-polytopes.
\newblock In A.~Bezdek, editor, {\em Discrete Geometry:~In Honor of W.
  Kuperberg's 60th Birthday}, Lecture Notes Pure Appl. Math., pages 267--278.
  Marcel-Dekker, 2003.

\bibitem[Epp00]{e-dtmcg-00}
\href{http://www.ics.uci.edu/~eppstein/}{D.~{Eppstein}}.
\newblock  Diameter and treewidth in minor-closed graph families.
\newblock {\em Algorithmica}, 27:275--291, 2000.

\bibitem[FG88]{fg-oafac-88}
T.~Feder and D.~H. Greene.
\newblock  Optimal algorithms for approximate clustering.
\newblock In {\em Proc. 20th Annu. ACM Sympos. Theory Comput. {\em(STOC)}},
  pages 434--444, 1988.

\bibitem[Fre87]{f-faspp-87}
G.~N. Frederickson.
\newblock  Fast algorithms for shortest paths in planar graphs, with
  applications.
\newblock {\em SIAM J. Comput.}, 16(6):1004--1022, 1987.

\bibitem[GKS14]{gks-dfopndg-14}
M.~Grohe, S.~Kreutzer, and S.~Siebertz.
\newblock  Deciding first-order properties of nowhere dense graphs.
\newblock In {\em Proc. 46th Annu. ACM Sympos. Theory Comput. {\em(STOC)}},
  pages 89--98, 2014.

\bibitem[Gro03]{g-ltwem-03}
M.~Grohe.
\newblock  Local tree-width, excluded minors, and approximation algorithms.
\newblock {\em Combinatorica}, 23(4):613--632, 2003.

\bibitem[{Har}09]{h-bffne-09}
\href{http://sarielhp.org}{S.~{{Har-Peled}}}.
\newblock  Being fat and friendly is not enough.
\newblock {\em CoRR}, 2009.

\bibitem[{Har}13]{h-speps-13}
\href{http://sarielhp.org}{S.~{{Har-Peled}}}.
\newblock  A simple proof of the existence of a planar separator.
\newblock {\em ArXiv e-prints}, April 2013.

\bibitem[{Har}14]{h-qssp-14}
\href{http://sarielhp.org}{S.~{{Har-Peled}}}.
\newblock  Quasi-polynomial time approximation scheme for sparse subsets of
  polygons.
\newblock In {\em Proc. 30th Annu. Sympos. Comput. Geom. {\em(SoCG)}}, pages
  120--129, 2014.

\bibitem[H{\aa}s99]{h-chaw-99}
Johan H{\aa}stad.
\newblock \href{http://dx.doi.org/10.1007/BF02392825}{Clique is hard to
  approximate withinn {$n^{1-\varepsilon}$}}.
\newblock {\em Acta Mathematica}, 182(1):105--142, 1999.

\bibitem[HKRS97]{hkrs-fspap-97}
M.~R. Henzinger, P.~Klein, S.~Rao, and S.~Subramanian.
\newblock  Faster shortest-path algorithms for planar graphs.
\newblock {\em J. Comput. Sys. Sci.}, 55:3--23, August 1997.

\bibitem[HQ15]{hq-aapel-15}
\href{http://sarielhp.org}{S.~{{Har-Peled}}} and K.~Quanrud.
\newblock  Approximation algorithms for polynomial-expansion and low-density
  graphs.
\newblock In {\em Proc. 23nd Annu. Euro. Sympos. Alg.\CNFESA}, volume 9294 of
  {\em Lect. Notes in Comp. Sci.}, pages 717--728, 2015.

\bibitem[HQ16]{hq-naape-16-arxiv}
\href{http://sarielhp.org}{S.~{{Har-Peled}}} and K.~{Quanrud}.
\newblock \href{http://arxiv.org/abs/1603.03098}{Notes on approximation
  algorithms for polynomial-expansion and low-density graphs}.
\newblock {\em ArXiv e-prints}, March 2016.

\bibitem[HR13]{hr-nplta-13}
\href{http://sarielhp.org}{S.~{{Har-Peled}}} and B.~Raichel.
\newblock \href{http://cs.uiuc.edu/~sariel/papers/12/aggregate/}{Net and prune:
  A linear time algorithm for {Euclidean} distance problems}.
\newblock In {\em Proc. 45th Annu. ACM Sympos. Theory Comput. {\em(STOC)}},
  pages 605--614, New York, NY, USA, 2013. ACM.

\bibitem[HT74]{ht-ept-74}
J.~E. Hopcroft and R.~E. Tarjan.
\newblock  Efficient planarity testing.
\newblock {\em \href{http://www.acm.org/jacm/}{J. Assoc. Comput. {Mach.}}}, 21(4):549--568, 1974.

\bibitem[IP01]{ip-ocks-01}
R.~Impagliazzo and R.~Paturi.
\newblock  On the complexity of {$k$}-{SAT}.
\newblock {\em J. Comput. Sys. Sci.}, 62(2):367--375, 2001.

\bibitem[IPZ01]{ipz-wphse-01}
R.~Impagliazzo, R.~Paturi, and F.~Zane.
\newblock  Which problems have strongly exponential complexity?
\newblock {\em J. Comput. Sys. Sci.}, 63(4):512--530, 2001.

\bibitem[Kar72]{k-racp-72}
R.~M. Karp.
\newblock
  \href{http://www.cs.berkeley.edu/$\sim$luca/cs172/karp.pdf}{Reducibility
  among combinatorial problems}.
\newblock In {\em Complexity of Computer Computations}, pages 85--103, 1972.

\bibitem[Koe36]{k-kdka-36}
P.~Koebe.
\newblock  Kontaktprobleme der konformen {Abbildung}.
\newblock {\em Ber. Verh. S{\"a}chs. Akademie der Wissenschaften Leipzig,
  Math.-Phys. Klasse}, 88:141--164, 1936.

\bibitem[LT79]{lt-stpg-79}
R.~J. Lipton and R.~E. Tarjan.
\newblock  A separator theorem for planar graphs.
\newblock {\em SIAM J. Appl. Math.}, 36:177--189, 1979.

\bibitem[LT80]{lt-apst-80}
R.~J. Lipton and R.~E. Tarjan.
\newblock  Applications of a planar separator theorem.
\newblock {\em SIAM J. Comput.}, 9(3):615--627, 1980.

\bibitem[Mat14]{m-nossg-14}
\href{http://kam.mff.cuni.cz/~matousek}{J. Matou{\v s}ek}.
\newblock  Near-optimal separators in string graphs.
\newblock {\em Combin., Prob. {\&} Comput.}, 23(1):135--139, 2014.

\bibitem[MR10]{mr-irghs-10}
Nabil~H. Mustafa and Saurabh Ray.
\newblock  Improved results on geometric hitting set problems.
\newblock {\em \href{http://link.springer.com/journal/454}{Discrete Comput. {}Geom.}}, 44(4):883--895, 2010.

\bibitem[MRR14a]{mrr-qgscp-14}
N.~H. {Mustafa}, R.~{Raman}, and S.~{Ray}.
\newblock \href{http://arxiv.org/abs/1403.0835}{{QPTAS} for geometric set-cover
  problems via optimal separators}.
\newblock {\em ArXiv e-prints}, 2014.

\bibitem[MRR14b]{mrr-sahsg-14}
N.~H. Mustafa, R.~Raman, and S.~Ray.
\newblock  Settling the {APX}-hardness status for geometric set cover.
\newblock In {\em Proc. 55th Annu. IEEE Sympos. Found. Comput. Sci.
  {\em(FOCS)}}, pages 541--550, 2014.

\bibitem[MTTV97]{mttv-sspnng-97}
G.~L. Miller, S.~H. Teng, W.~P. Thurston, and S.~A. Vavasis.
\newblock  Separators for sphere-packings and nearest neighbor graphs.
\newblock {\em \href{http://www.acm.org/jacm/}{J. Assoc. Comput. {Mach.}}}, 44(1):1--29, 1997.

\bibitem[NO08a]{no-gcbe1-08}
J.~{Ne{\v s}et{\v r}il} and P.~{Ossona de Mendez}.
\newblock  Grad and classes with bounded expansion {I}. decompositions.
\newblock {\em Eur. J. Comb.}, 29(3):760--776, 2008.

\bibitem[NO08b]{no-gcbe2-08}
J.~{Ne{\v s}et{\v r}il} and P.~{Ossona de Mendez}.
\newblock  Grad and classes with bounded expansion {II}. algorithmic aspects.
\newblock {\em Eur. J. Comb.}, 29(3):777--791, 2008.

\bibitem[NO12]{no-s-12}
J.~Ne{\v s}et{\v r}il and P.~{Ossona de Mendez}.
\newblock  {\em Sparsity -- Graphs, Structures, and Algorithms}, volume~28 of
  {\em Alg. Combin.}
\newblock Springer, 2012.

\bibitem[PA95]{pa-cg-95}
\href{http://www.math.nyu.edu/~pach}{J.~{Pach}} and \href{http://www.cs.duke.edu/~pankaj}{P.~K.~{Agarwal}}.
\newblock \href{http://www.addall.com/Browse/Detail/0471588903.html}{{\em
  Combinatorial Geometry}}.
\newblock John Wiley \& Sons, 1995.

\bibitem[RS97]{rs-sbepl-97}
R.~Raz and S.~Safra.
\newblock  A sub-constant error-probability low-degree test, and a sub-constant
  error-probability {PCP} characterization of {NP}.
\newblock In {\em Proc. 29th Annu. ACM Sympos. Theory Comput. {\em(STOC)}},
  pages 475--484, 1997.

\bibitem[Sei91]{s-eubnf-91}
\href{http://www-tcs.cs.uni-sb.de/seidel/}{R.~{Seidel}}.
\newblock \href{http://dimacs.rutgers.edu/Volumes/Vol04.html}{Exact upper
  bounds for the number of faces in {$d$}-dimensional {V}oronoi diagrams}.
\newblock In P.~Gritzman and B.~Sturmfels, editors, {\em Applied Geometry and
  Discrete Mathematics: The Victor Klee Festschrift}, volume~4 of {\em DIMACS
  Series in Discrete Mathematics and Theoretical Computer Science}, pages
  517--530. Amer. Math. Soc., 1991.

\bibitem[SS85]{ss-empae-85}
J.~T. Schwartz and \href{http://www.math.tau.ac.il/~michas}{M.~{Sharir}}.
\newblock  Efficient motion planning algorithms in environments of bounded
  local complexity.
\newblock Report 164, Dept. Comput. Sci., Courant Inst. Math. Sci., New York
  Univ., New York, NY, 1985.

\bibitem[{Sta}92]{f-mpafo-92}
{A. F. van~der} {Stappen}.
\newblock
  \href{http://www.staff.science.uu.nl/~stapp101/PhDThesis_AFvanderStappen.pdf}{{\em
  Motion Planning Amidst Fat Obstacles}}.
\newblock PhD thesis, Utrecht University, Netherlands, 1992.

\bibitem[SW98]{sw-gsta-98}
W.~D. Smith and N.~C. Wormald.
\newblock  Geometric separator theorems and applications.
\newblock In {\em Proc. 39th Annu. IEEE Sympos. Found. Comput. Sci.
  {\em(FOCS)}}, pages 232--243, 1998.

\bibitem[Tie05]{t-upnr-05}
H.~Tietze.
\newblock  Uber das problem der nachbargeibiete im raum.
\newblock {\em Monatshefte Math.}, 15:211--216, 1905.

\bibitem[vdSOdBV98]{sobv-mpelo-98}
A.~F. van~der Stappen, \href{http://www.cs.uu.nl/people/markov/}{M.~H. {Overmars}}, \href{http://www.win.tue.nl/~mdberg/}{M.~de~{Berg}}, and J.~Vleugels.
\newblock  Motion planning in environments with low obstacle density.
\newblock {\em \href{http://link.springer.com/journal/454}{Discrete Comput. {}Geom.}}, 20(4):561--587, 1998.

\bibitem[{Ver}05]{v-cbseb-05}
J.-L. {Verger-Gaugry}.
\newblock  Covering a ball with smaller equal balls in {$\Re^n$}.
\newblock {\em \href{http://link.springer.com/journal/454}{Discrete Comput. {}Geom.}}, 33(1):143--155, 2005.

\end{thebibliography}




\end{document}